\documentclass[]{amsart}
\usepackage[utf8]{inputenc}
\usepackage[a4paper, margin=1in]{geometry}
\usepackage{bbm,amsmath,amsfonts,amssymb,amsthm,amscd,amsbsy,mathabx,epsfig,array,mathrsfs,tikz,tikz-cd,url,graphicx,subcaption,csquotes}
\usepackage[colorlinks=false,citecolor=green,linktoc=page]{hyperref}

\usepackage{tikz}
\usetikzlibrary{trees,calc,decorations.markings,decorations.pathmorphing}

\theoremstyle{plain}
\newtheorem{thm}{Theorem}[section]
\newtheorem*{thm*}{\bf Theorem }
\newtheorem{lem}[thm]{Lemma}
\newtheorem{prop}[thm]{Proposition}
\newtheorem*{prop*}{\bf Proposition}
\newtheorem{cor}[thm]{Corollary}

\theoremstyle{definition}
\newtheorem{defn}[thm]{Definition}

\newtheorem{example}[thm]{Example}

\theoremstyle{remark}

\numberwithin{equation}{section}

\title[Momentum Space Landau Equations Via Isotopy Techniques]{Momentum Space Landau Equations Via Isotopy Techniques}

\author{Maximilian M{\"u}hlbauer}

\begin{document}
\tikzset{
    photon/.style={decorate, decoration={snake}, draw=black},
    electron/.style={draw=black, postaction={decorate},
        decoration={markings,mark=at position .55 with {\arrow[draw=black]{>}}}},
    gluon/.style={decorate, draw=black,
        decoration={coil,amplitude=4pt, segment length=5pt}}
}

\begin{abstract}
    We investigate the analytic structure of functions defined by integrals with integrands singular on a finite union of quadrics. The main motivation comes from Feynman integrals which belong to this class. Using isotopy techniques we derive the Landau equations in momentum space from the theory of Feynman integrals and generalize these equations to naturally include singularities of the second type. For this purpose we introduce a regularization of analytic families of quadratic forms rendering the isotopy techniques applicable. In the case of Feynman integrals we comment on what is known about the behavior on the principal branch where only specific solutions of the Landau equations contribute to non-analytic points. Finally we discuss compatibility with renormalization.
\end{abstract}

\maketitle
\tableofcontents

\section{Notation and Conventions}
We agree on the following conventions: For a set $X$ and subsets $A,B\subset X$ we denote the relative complement of $A$ in $B$ by $B\backslash A$ or $B-A$. If $Y$ is another set and $f:X\to Y$ is a map we denote its image by $\text{im}\,f\subset Y$. The set of natural numbers including 0 is denoted by $\mathbb{N}$ and we set $\mathbb{N}^+:=\mathbb{N}\backslash\{0\}$. Similarly we write $\mathbb{R}^+$ for the set of positive real numbers, $\mathbb{R}_{\geq0}$ for the set of non-negative real numbers. The complex numbers without 0 are denoted by $\mathbb{C}^\times:=\mathbb{C}\backslash\{0\}$. For a complex number $z\in\mathbb{C}$ we denote its real part by $\text{Re}(z)$ and its imaginary part by $\text{Im}(z)$. For a tuple of complex numbers $z=(z_1,\ldots,z_n)\in\mathbb{C}^n$ we define
\begin{align*}
    \text{Re}(z):=(\text{Re}(z_1),\ldots,\text{Re}(z_n))\qquad\text{and}\qquad\text{Im}(z):=(\text{Im}(z_1),\ldots,\text{Im}(z_n)).
\end{align*}
Let $X$ and $Y$ be topological spaces and $A\subset X$ a subset. We always equip $X\times Y$ with the product topology and $A$ with the subspace topology. We denote the closure of $A$ in $X$ by $\text{cl}(A)$. Furthermore if $Z$ is a third topological space we speak of continuous families of maps $\{f_z:X\to Y\}_{z\in Z}$ which shall mean that the map $X\times Z\to Y,\;(x,z)\mapsto f_z(x)$ is continuous. We employ a similar jargon in the case of differentiable or complex analytic manifolds and speak of differentiable or (complex) analytic families of maps. If we say something holds almost everywhere or almost nowhere we always refer to the Lebesque measure on $\mathbb{C}^n$ or $\mathbb{R}^n$.

\section{Introduction}
The aim of this paper is to investigate the analytical structure of functions defined by integrals of the form
\begin{align}\label{eq:introintegral}
    I:U\to\mathbb{C},\quad t\mapsto\int_{\mathbb{R}^n}\frac{p(t)(y)}{\prod_{i=1}^N(Q_i(t)(y))^{\lambda_i}}dy_1\wedge\cdots\wedge dy_n
\end{align}
where $U\subset\mathbb{C}^m$ is an open set, the $\{Q_i(t)\}_{t\in U}$ are families of quadratic functions on $\mathbb{C}^n$ depending analytically on $t\in U$, $\{p(t)\}_{t\in U}$ is a family of polynomials analytically depending on $t$, and $\lambda_1,\ldots,\lambda_N\in\mathbb{C}$ with $\text{Re}(\lambda_i)>0$ for all $i\in\{1,\ldots,N\}$. We choose this notation to consistently distinguish between the variables over which we are integrating and the parameters on which $f$ depends. The main motivation lies in understanding the analytic structure of functions defined by Feynman integrals (in the momentum space representation) in quantum field theory which belong to this class of integrals.\\
Even though these functions have been heavily investigated, leading to astonishing connections to number theory \cite{periods}, elliptic curves \cite{elliptic-dilog}, tropical geometry \cite{hepp-bound} and many more flourishing areas of contemporary mathematics, many features of such functions are not yet well-understood. A famous theorem due to Cutkosky (see \cite{cutkosky}) for example, relating the monodromy around singular points to an integral with certain edges \text{on mass-shell}, remains unproven to this day (Which is to say, there is no published paper containing a rigorous proof of the theorem. The article \cite{outer-space}, available on the arXiv, attempting a proof is, as far as the author understands, work in progress.) although the statement dates back to the 60s. Closer inspection of the literature reveals that there seem to be gaps regarding a rigorous derivation of the Landau equations, characterizing the points at which singularities can occur, in the momentum space setting. The fact that clarification is needed on various points is illustrated by how recently the paper \cite{landaucollins} by Collins appeared, containing independent work dedicated to a rigorous proof of the Landau equations and refining the known results. This paper also showcases potential problems with the parametric viewpoint: It is not clear that the momentum space integral and the parametric integral share the same pinch points in all kinematic situations and in fact \cite{landaucollins} contains an example showing that this is not true in general. Problems occur in particular for massless theories. Results on the Landau equations for the parametric representation of Feynman integrals can be found for example in \cite{periods} or \cite{konrad}. Since Cutkosky's Theorem is formulated in momentum space, this gap should be closed and this paper attempts to do so. Therefore we work exclusively in momentum space and refrain from employing the Feynman or the Schwinger trick.\\
It was suggested (e.g. in \cite{homfeynman} and \cite{s-matrix}) that the isotopy techniques developed in \cite{app-iso}, which can be used in particular to deal with integrals of a specific form (see Definition \ref{def:standard_form} below) and are often cited in the literature regarding the Landau equations, might be applied to the case of Feynman integrals in momentum space representation after suitable modifications. The main goal of this work is to make concrete which modifications are needed and to apply them explicitly. The integrals studied in \cite{app-iso} are defined as follows:
\begin{defn}[\cite{app-iso}]\label{def:standard_form}
    Let $Z$ and $T$ be two complex analytic manifolds with $\dim Z=n$. An integral
    \begin{align*}
        \int_{\Gamma_0}\frac{\omega_{t_0}}{s(t_0)}
    \end{align*}
    is said to be \textit{of standard form} if the following three conditions are fulfilled:
    \begin{enumerate}
        \item $S(t):=\{t\in T \;|\; s(t)=0\}$ is a closed analytical subset of $Z$, depending analytically on $t\in T$ (in particular $s(t):Z\to\mathbb{C}$ is an analytic function for all $t\in T$).
        \item $\omega_t$ is a regular analytic $n$-form on $Z$, depending analytically on $t\in T$.
        \item $\Gamma_0\subset Z-S(t_0)$ defines a compact $n$-cycle.
    \end{enumerate}
\end{defn}
The word \enquote{regular} is somewhat superfluous (and quoted from the original paper). It is simply intended to stress that all singularities of the integrand are encoded in $s(t)$. An integral of standard form defines a germ of an analytic function at $t\in T$ if there is a path $\gamma:[0,1]\to T$ from $t_0$ to $t$ such that the fiber bundle of pairs $\pi:(X\times T,S)\to T$ (see Definition \ref{def:fibre_pair}) is locally trivial at every $t\in\text{im}\,\gamma$ as we discuss below. We review the necessary details of this procedure in Section \ref{sec:preliminaries}.\\
The authors who developed these techniques attempted to apply them to Feynman integrals and seem to have succeeded in the one-loop case (which is very special from the isotopy perspective). The two papers, however, containing explicit calculations for the box and the kite diagram respectively, have never been published. A draft of these works can be found in the book \cite{homfeynman} but unfortunately the second paper regarding the kite integral is largely incomplete.\\
The structure of this paper is as follows: The next Section \ref{sec:preliminaries} reviews the relevant definitions and theorems from the complex analysis in several variables and from the isotopy techniques in \cite{app-iso}. The main tool for us is Corollary \ref{thm:main_corollary}, which allows us to check the analyticity of an integral of standard form by a simple condition. Section \ref{sec:quadric_integrals} contains the main discussion introducing the relevant definitions and a regularization of quadratic forms and integrals in the class under investigation to allow us to apply the aforementioned corollary directly. We obtain a (slight) generalization of what physicists call the Landau equations naturally including singularities \enquote{of the second type}. The following Section \ref{sec:feynman_integral} contains the details for the case of Feynman integrals including a quick review of the general theory. We remark on the restriction of the solution space of the Landau equation if one is interested in the analyticity at specific (physical) points on the principal branch. To understand why it suffices to consider a smaller space of solutions it is necessary to perform an analytic continuation from $t_0$, corresponding to Euclidean external momenta and real masses in the case of Feynman integrals, to the value $t\in T$ of interest. This step has not been well-understood so far, as remarked for example in \cite{hyperlog}. We briefly discuss the role which renormalization plays in this picture in Section \ref{sec:renormalization}. The concluding Section \ref{sec:examples} contains a small list of examples both from physics and pure mathematics to illustrate how these techniques work in practice.

\section{Preliminaries}\label{sec:preliminaries}
We review some basics from the analysis in several complex variables and the techniques developed in \cite{app-iso} for the reader's convenience. This section also functions to fix some of the relevant notation. The part on complex analysis is mostly standard and is available for example in \cite{ana-book} and \cite{chirka} while most of what is contained in Subsection \ref{sec:isotopy} is covered in the book \cite{pham} by Pham or can be found in the original paper \cite{app-iso}. Everything not contained in these sources will be explicitly mentioned and supplied with additional references.
\subsection{Recap: Complex Analysis in Several Variables}
We assume the reader is familiar with the basic notions of complex analysis in several variables. Since the notions of being a holomorphic or an analytic function coincide in complex analysis we stick to the term analytic (which will always mean complex analytic) throughout this work. For an open set $U\subset\mathbb{C}^n$ we denote the space of analytic functions $U\to\mathbb{C}$ by $\mathcal{O}(U)$. As an algebraic structure $\mathcal{O}(U)$ is a $\mathbb{C}$-algebra. We endow it with the compact-open topology which is induced by a family of semi-norms $\{\rho_{K_i}\}_{i\in\mathbb{N}}$ where $\{K_i\}_{i\in\mathbb{N}}$ is a fixed compact exhaustion of $U$ and $\rho_K$ is defined by
\begin{align*}
    \rho_K:\mathcal{O}(U)\to\mathbb{R}_{\geq0},\quad f\mapsto \text{sup}_{z\in K}|f(z)|
\end{align*}
for any compact set $K\subset U$. This family induces a metric
\begin{align*}
    d:\mathcal{O}(U)\times\mathcal{O}(U)\to\mathbb{R}_{\geq0},\quad (f,g)\mapsto\sum_{i=1}^\infty\frac{1}{2^i}\frac{\rho_{K_i}(f-g)}{1+\rho_{K_i}(f-g)}
\end{align*}
turning $\mathcal{O}(U)$ into a Fréchet algebra. In particular we use the following proposition:
\begin{prop}[\cite{ana-book}]\label{thm:compact_convergence}
    If $\{f_k\}_{k\in\mathbb{N}}$ is a sequence of analytic functions on $U\subset\mathbb{C}^n$ converging uniformly on every compact subset $K\subset U$ then $\lim_{k\to\infty}f_k:U\to\mathbb{C}$ defines an analytic function.
\end{prop}
Analogous to algebraic sets in algebraic geometry, a special role is played by sets locally given as the common zeros of families of analytic functions.
\begin{defn}[\cite{ana-book}]
    Let $U\subset\mathbb{C}^n$ be an open set and $A\subset U$. The set $A$ is \textit{analytic} in $U$ if for every $a\in A$ there is a neighborhood $V\subset U$ of $a$ and a finite set of functions $f_1\ldots,f_m\in\mathcal{O}(V)$ such that
    \begin{align*}
        A\cap V=\{z\in V \;|\; f_1(z)=\cdots=f_m(z)=0\}.
    \end{align*}
\end{defn}
This definition naturally extends to complex manifolds since it is local. Clearly $\mathbb{C}^n$ and $\emptyset$ are analytic sets as they are they are (globally) given by $\{z\in\mathbb{C}^n \;|\; f(z)=0\}$ for $f\equiv0:\mathbb{C}^n\to\mathbb{C}$ and $f\equiv1:\mathbb{C}^n\to\mathbb{C}$ respectively. Furthermore analytic sets are closed under finite unions and arbitrary intersections \cite{chirka} and can thus be used to define the closed sets of a topology on $\mathbb{C}^n$ similarly to the Zarisky topology. For an open set $U\subset\mathbb{C}^n$ and a subset $S\subset\mathcal{O}(U)$ we write
\begin{align*}
    V(S):=\{z\in U \;|\; f(z)=0\quad\forall f\in S\}
\end{align*}
for the \textit{zero locus} of $S$. In particular if $S$ contains only one element $f\in\mathcal{O}(U)$ we write $V(f):=V(\{f\})$. In the case of homogeneous functions we are often interested in the zero locus in projective space in the sense of $\{[z]\in\mathbb{C}\mathbb{P}^n \;|\; f(z)=0\quad\forall f\in S\}$ which we denote by the same symbol. It will always be clear from context which meaning we invoke. Let $A$ be an analytic set and $k\in\mathbb{N}$. We say that $A\subset\mathbb{C}^n$ has \textit{codimension} $k$ at $a\in A$, denoted by $\text{codim}_aA=k$, if there is an affine subspace $E$ of $\mathbb{C}^n$ of dimension $k$ such that $a$ is an isolated point in $E\cap A$ and there is no affine subspace of larger dimension than $k$ with this property. If $A$ is non-empty we define the \textit{codimension} of $A$ by $\text{codim}A:=\inf_{a\in A}\text{codim}_aA$. We frequently use the following elementary properties of analytic sets:
\begin{prop}[\cite{ana-book}]\label{thm:basics_analytic_sets}
    Let $U\subset\mathbb{C}^n$ be an open set, $A\subset U$ an analytic set. Then the following hold:
    \begin{enumerate}
        \item If $U$ is connected then $A\neq U$ is equivalent to $\textup{codim}\,A\geq1$.
        \item If $\textup{codim}\,A\geq1$ then $U\backslash A$ is dense in $U$.
    \end{enumerate}
\end{prop}
Note that if $A\subset U$ is analytic then $A$ is closed and hence (2) from the above proposition tells us that $\text{codim}\,A\geq1$ is equivalent to $A$ being nowhere dense in $U$. Additionally we employ the Proper Mapping Theorem due to Remmert. Recall that a proper map $f:X\to Y$ between two topological spaces $X$ and $Y$ is a map such that for every compact set $K\subset Y$ the preimage $f^{-1}(K)$ is compact in $X$.
\begin{thm}[Remmert's Proper Mapping Theorem, \cite{chirka}]\label{thm:proper_mapping_thm}
    Let $U\subset\mathbb{C}^m$ and $V\subset\mathbb{C}^n$ be open sets, $f:U\to V$ a proper analytic map. If a subset $A\subset U$ is analytic in $U$ then $f(A)$ is analytic in $V$.
\end{thm}
An elementary but powerful tool we often use is the principal of analytic continuation. The following well-known theorem is a very useful criterion to determine if an analytic function in several variables can be analytically continued to a larger domain. If $U\subset\mathbb{C}^m$ and $V\subset\mathbb{C}^n$ are open subsets and $A\subset U$ is any subset we say a map $f:U\backslash A\to V$ is \textit{bounded along} $A$ if every point $a\in A$ has a neighborhood $N_a$ in $U$ such that $f$ is bounded on $N_a-A$ \cite{ana-book}.
\begin{thm}[Riemann's Extension Theorem, \cite{ana-book}]\label{thm:riemann_extension}
    Let $U\subset\mathbb{C}^n$ be open and $A\subset U$ analytic with $\textup{codim}\,A\geq1$. If $f:U\backslash A\to\mathbb{C}$ is an analytic function bounded along $A$ then $f$ can be extended to an analytic function $U\to\mathbb{C}$.
\end{thm}
This theorem has a useful corollary about the topological nature of complements of analytic sets:
\begin{cor}[\cite{ana-book}]\label{thm:connected_corollary}
    Let $U\subset\mathbb{C}^n$ be an open and connected set. If $A\subset U$ is analytic with $\textup{codim}\,A\geq1$ then $U\backslash A$ is connected.
\end{cor}
\begin{proof}
    Suppose we can write $U\backslash A=V_1\cup V_2$ with $V_1\cap V_2=\emptyset$ and $V_1,V_2\neq\emptyset$. Then we can define a function
    \begin{align*}
        f:U\backslash A\to\mathbb{C},\quad z\mapsto\begin{cases} 1 & \text{if } z\in V_1 \\ 0 & \text{if } z\in V_2 \end{cases}
    \end{align*}
    which is clearly analytic. Moreover it is certainly bounded along $A$ so that we can analytically continue to all of $U$. But $U$ is connected and $f$ locally constant so that $f\equiv0$ and $f\equiv1$ which is absurd.
\end{proof}
Recall that any connected and locally path-connected topological space is path-connected. This holds in particular for open subsets of $\mathbb{R}^n$ or $\mathbb{C}^n$. Thus the situation in Corollary \ref{thm:connected_corollary} already implies path-connectedness. Another consequence of Theorem \ref{thm:riemann_extension} is the last theorem of this subsection which is sometimes called Riemann's Second Extension Theorem.
\begin{thm}[Riemann's 2nd Extension Theorem, \cite{ana-book}]\label{thm:riemann_extension_2}
    Let $U\subset\mathbb{C}^n$ be open and $A\subset U$ analytic with $\textup{codim}\,A\geq2$. If $f:U\backslash A\to\mathbb{C}$ is an analytic function then $f$ can be extended to an analytic function $U\to\mathbb{C}$.
\end{thm}

\subsection{Isotopy Techniques}\label{sec:isotopy}
We review the techniques developed in \cite{app-iso} to treat integrals of standard form (see Definition \ref{def:standard_form}). First we need the notion of a fiber bundle of pairs and certain compatibility with the differentiable structure in case the topological spaces under consideration are differentiable manifolds. Smooth always means $C^\infty$ and we restrict ourselves to smooth (i.e. $C^\infty$-) manifolds in all that follows.
\begin{defn}[\cite{pham}]\label{def:fibre_pair}
    Let $Y$ and $T$ be topological spaces, $S\subset Y$ a subspace and $\pi:Y\to T$ a continuous map. We say that $\pi:(Y,S)\to T$ is a \textit{(locally) trivial fiber bundle of pairs} if there exists a (local) trivialization of $Y$ which is also a (local) trivialization for $S$ (generally with a different fiber).
\end{defn}
Note that we explicitly state \enquote{locally trivial} in this definition as opposed to the definition of fiber bundles which usually already includes local triviality (although this depends on the author). This convention is adopted from \cite{pham}. In case of bundles of smooth manifolds we want to impose some extra conditions:
\begin{defn}[\cite{app-iso}]\label{def:smoothly_trivial}
    Let $Y$ and $T$ be smooth manifolds, $S\subset Y$ a topological subspace. A trivial fiber bundle of pairs $\pi:(Y,S)\to T$ is called a \textit{trivial $\mathcal{C}^\infty$-fiber bundle of pairs} if the following conditions hold:
    \begin{enumerate}
     \item There exists a trivialization $g:Y\to Z\times T$ which also trivializes $S$ (i.e $\pi:(Y,S)\to T$ is a trivial fiber bundle of pairs) and the fiber $Z$ is a smooth manifold.
     \item $\pi$ is a smooth map.
     \item The inverse $g^{-1}$ is differentiable with respect to $t\in T$.
     \item Any smooth vector field $F$ on $T$ lifts to a locally Lipschitzian vector field $g^{-1}F$ on $Y$.
    \end{enumerate}
    Furthermore we say that $\pi:(Y,S)\to T$ is a \textit{locally trivial $\mathcal{C}^\infty$-fiber bundle of pairs} if each $t\in T$ has an open neighborhood $U\subset T$ such that
    \begin{align*}
        \pi:(\pi^{-1}(U),S\cap\pi^{-1}(U))\to U
    \end{align*}
    is a trivial $\mathcal{C}^\infty$-fiber bundle of pairs.
\end{defn}
While the first three conditions seem intuitive, the fourth one warrants some explanation. First recall the definition of the lifted vector field $g^{-1}F$ given by $(g^{-1}F)(y):=D_t(g^{-1}(z,\cdot))F(t)$ for any point $y\in Y$ where $(z,t)=g^{-1}(y)$. This is well-defined due to condition (3) in Definition \ref{def:smoothly_trivial}. Note also that (4) is a strictly weaker condition than requiring $g^{-1}$ to be continuously differentiable with respect to $z\in Z$ (any continuously differentiable map is locally Lipschitz continuous). We will see further below that condition (4) is enough for the proof of the following isotopy theorem, which lies at the heart of \cite{app-iso}:
\begin{thm}[\cite{app-iso}]\label{thm:isotopy_thm}
    Let $Y$ and $T$ be smooth manifolds, $\pi:Y\to T$ a smooth, open, and proper map, and $S\subset Y$ an open or closed subset. Then the following are equivalent:
    \begin{enumerate}
        \item $\pi:(Y,S)\to T$ is a locally trivial $\mathcal{C}^\infty$-fiber bundle of pairs.
        \item For any $y\in Y$ there exists $U\subset Y$ open such that $\pi:(U,S\cap U)\to T$ is a locally trivial $\mathcal{C}^\infty$-fiber bundle of pairs.
    \end{enumerate}
\end{thm}
What this means is that local triviality in the $\mathcal{C}^\infty$-sense, which is a notion local in $T$, is equivalent to the same condition when we consider it also locally in $Y$.\\
Recall that given a smooth manifold $X$, any smooth submanifold $S\subset X$ of codimension $n$ can be characterized by $n$ equations (see for example \cite{pham}): For every $x\in S$ there is a neighborhood $U\subset X$ of $y$ and a smooth map $s:U\to\mathbb{R}^n$ of rank $n$ such that
\begin{align*}
    S\cap U=\{x\in U\;|\; s(x)=0\}.
\end{align*}
We call the $n$ components $s_i:U\to\mathbb{R}$ of $s=(s_1,\ldots,s_n)$ a set of \textit{local equations} for $S$. Now let $\{S_i\}_{i=1}^m$ be a family of smooth submanifolds of $X$, $x\in\bigcup_{i=1}^mS_i$ a point, and for each $i\in\{1,\ldots m\}$ let $s_1^{(i)},\ldots,s_{n_i}^{(i)}:U\to\mathbb{R}$ be local equations for $S_i$ at $x$. Denote
\begin{align*}
    J=\{j_1,\ldots,j_k\}:=\{i\in\{1,\ldots,m\} \;|\; x\in S_i\}.
\end{align*}
We say the family $\{S_i\}_{i=1}^m$ is in \textit{general position} at $x$ if there exists a chart $(U,\varphi)$ around $x$ such that the gradients of
\begin{align*}
    s_1^{(j_1)}\circ\varphi^{-1},\ldots,s_{n_{j_1}}^{(j_1)}\circ\varphi^{-1},s_1^{(j_2)}\circ\varphi^{-1},\ldots,s_{n_{j_k}}^{(j_k)}\circ\varphi^{-1}
\end{align*}
at $\varphi(x)$ are linearly independent over $\mathbb{R}$. It is not difficult to see that this definition does not depend on the choice of local equations since two sets local equations for a submanifold are connected by a diffeomorphism. If the family $\{S_i\}_{i\in I}$ is in general position at every $x\in V$ for some subset $V\subset\bigcup_{i\in I}S_i$ we say the family is in \textit{general position on $V$}, or simply in \textit{general position} if $V=\bigcup_{i\in I}S_i$. We will always use Theorem \ref{thm:isotopy_thm} in form of the following corollary:
\begin{cor}[\cite{app-iso}]\label{thm:main_corollary}
    Let $X$ and $T$ be smooth manifolds with $X$ compact. Let $\{S_i(t)\}_{i\in I}$ be a family of closed smooth submanifolds of $X$ which depend smoothly on $t\in T$ and are in general position for all $t\in T$. Denote $S(t):=\bigcup_{i\in I}S_i(t)$. Then $\pi:(X\times T,\bigcup_{t\in T}S(t))\to T$ is a locally trivial fiber bundle of pairs, where $\pi:X\times T\twoheadrightarrow T$ is the canonical projection.
\end{cor}
We will not prove Theorem \ref{thm:isotopy_thm} here (the reader is referred to the original paper for a detailed proof), among other reasons since we do not explicitly reference the $\mathcal{C}^\infty$-triviality from the theorem in what follows, but only use it implicitly by applying Corollary \ref{thm:main_corollary}. But we nevertheless give a short sketch of the proof to display the ideas involved.\\
The implication $(1)\Rightarrow(2)$ is trivial (take $U=Y$) and only $(2)\Rightarrow(1)$ warrants a proof. Any Lipschitzian vector field on a smooth manifold $X$ with compact support defines a one-parameter group of homeomorphisms $X\overset{\sim}{\to}X$. A superposition lemma (Lemma 2 in \cite{app-iso}) ensures that, given two such vector fields with corresponding groups of homeomorphisms leaving invariant a subset $S\subset X$, their sum corresponds to a one-parameter group of homeomorphisms leaving $S$ invariant as well (in fact even the closure of $S$ in $X$). By assumption the fibers $\pi^{-1}(\{t\})$ are compact for all $t\in T$, so each one of them can be covered by finitely many open subsets $U_1,\ldots,U_k\subset Y$ such that $\pi:(Y,U_j)\to T$ is locally trivial (in the $\mathcal{C}^\infty$-sense) for all $1\leq j\leq k$. We denote by $W(t)$ the smallest open neighborhood among the $\pi(U_i)$. If now $H$ is a unit vector field on $\mathbb{R}^{p}$ (with $p$ the real dimension of $T$) we can construct a vector field $fH$ on $W(t)$ by means of a bump-function $f$. By assumption $(4)$ these lift to locally Lipschitzian vector fields (which we may assume to be Lipschitzian by possibly restricting $W(t)$ to an even smaller neighborhood) on $Y$ with compact support. By using a partition of unity subordinate to the $U_i$ we obtain Lipschitzian vector fields on $Y$ with compact support on the $U_i$. The sum of these vector fields yields a vector field $F$ which then generates a one-parameter group of homeomorphisms leaving $S$ invariant. Choosing an open cube $K^p$ in $W(t)$ centered around $t$ and given by coordinates $x_1,\ldots,x_p$ decomposition of the cube then allows for induction on $p:=\dim T$: We choose $H$ such that it is parallel to the $x_p$-direction and decompose $K^p=K^{p-1}\times K^1$ where $K^{p-1}$ is determined by the vanishing of $x_p$, $K^1$ by the vanishing of the remaining coordinates. Since $Y$ and $S$ restricted to $K^{p-1}$ satisfy the conditions of the theorem we may apply the induction hypothesis to them. One can then construct from the local trivialization thus obtained a local trivialization for $Y$ and $S$ restricted to $K^p$ completing the proof.\\
Regarding the corollary first note that the canonical projection $\pi:X\times T\twoheadrightarrow T$ is open and is proper if and only if $X$ is compact. Now since the $S_i$ are smooth in $t$ around any $y\in Y$ there is a neighborhood $U\subset Y$ such that $S_i\cap U$ can be written as the zero set of finitely many $C^\infty$-function $s_1^{(i)},\ldots,s_{n_i}^{(i)}$. General position means that at a point $x_0$ the gradients of all $s_j^{(i)}$ belonging to the $S_i$ intersecting at $x_0$ are linearly independent. We may therefore define the sought trivialization by using the $s_j^{(i)}$ as coordinates (together with coordinates for $t$ and the remaining coordinates for $x$).

\subsection{Treatment of Integrals}\label{sec:integraltreatment}
The general idea to treat an integral of standard form with these techniques is roughly the following: At $t_0$ we have a well-defined integral, the integration domain is compact and does not meet the zero locus of the denominator. If we now try to move along a path $\gamma:[0,1]\to T$ away from $t_0$ to some other $t_1\in T$, the initial zero locus $S(t_0)$ changes continuously along the way. This might lead to a collision of the zero locus with $\Gamma_0$ which potentially causes problems. Invoking Stoke's Theorem however we know that the integral does only depend on the homology class of $\Gamma_0$ (in the ambient space minus the zero loci) so that we can deform it suitably to avoid this from happening. The manifold in which we are allowed to deform the cycle is called the \textit{ambient space} ($Z$ in Definition \ref{def:standard_form}). It is certainly not always possible to find a suitable deformation since the integration cycle might be trapped by the singular locus. This is depicted in Figure \ref{fig:pinch} for the simple example $\int_\sigma \frac{dz}{z^2-t}$ where $\sigma$ is a small circle around $1\in\mathbb{C}$.
\begin{figure}
    \centering
    \begin{subfigure}[b]{0.45\textwidth}
        \begin{tikzpicture}[scale=0.70]
        \def\gap{0.2}
        \def\bigradius{3}
        \def\littleradius{0.5}

        \draw [help lines,->] (-1.25*\bigradius, 0) -- (1.25*\bigradius,0);
        \draw [help lines,->] (0, -1.25*\bigradius) -- (0, 1.25*\bigradius);

        \draw (0.8,0) arc (-180:180:0.7);

        \node at (2.2,-0.8){$\sigma$};
        \node at (3.6,-0.3){$x$};
        \node at (-0.35,3.53) {$iy$};
        \node at (1.5,0) {$\times$};
        \node at (-1.5,0) {$\times$};
        \node at (1.5,-0.45) {$1$};
        \node at (-1.5,-0.45) {$-1$};
        \draw [->] (1.5,0) to (1,0);
        \draw [->] (-1.5,0) to (-1,0);
        \end{tikzpicture}
        \caption{$t=1$}
        \label{fig:ex1}
    \end{subfigure}
    ~ 
    \begin{subfigure}[b]{0.45\textwidth}
        \begin{tikzpicture}[scale=0.70]
        \def\gap{0.2}
        \def\bigradius{3}
        \def\littleradius{0.5}

        \draw [help lines,->] (-1.25*\bigradius, 0) -- (1.25*\bigradius,0);
        \draw [help lines,->] (0, -1.25*\bigradius) -- (0, 1.25*\bigradius);

        \draw (2,0) to [out=90, in=90] (0,0);
        \draw (2,0) to [out=-90, in=-90] (0,0);
        \node at (3.6,-0.3){$x$};
        \node at (-0.35,3.53) {$iy$};
        \node at (0.15, 0) {$\times$};
        \node at (-0.15, 0) {$\times$};
        \end{tikzpicture}
        \caption{$t\to0$}
        \label{fig:tiger}
    \end{subfigure}
    \caption{The situation for the integral $\int_\sigma \frac{dz}{z^2-t}$: The integration contour $\sigma$ can be deformed to avoid the singular points of the integrand, except at $t=0$, where the singular points pinch the contour.}\label{fig:pinch}
\end{figure}
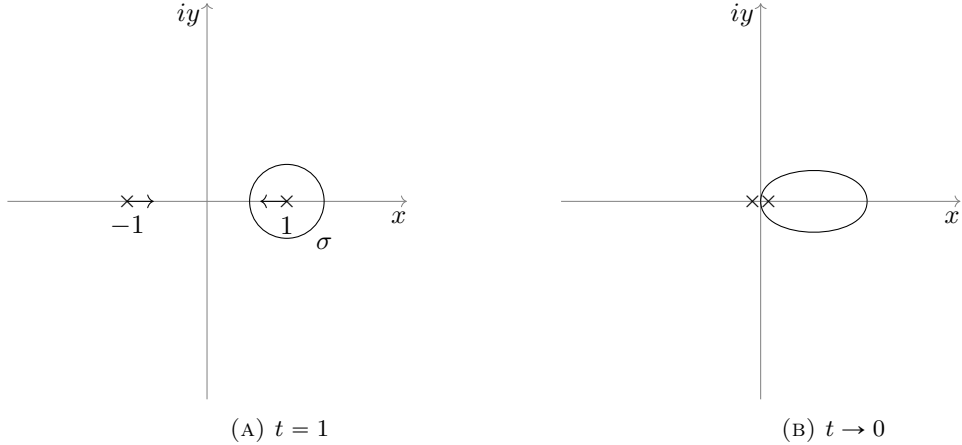
But assuming local triviality a path from $t_0$ to $t_1$ generates a continuous family of homeomorphisms leaving invariant the zero locus and yielding the desired deformation. To make this idea precise we need the following notions:
\begin{defn}[\cite{pham}]
    Let $Z$ be a manifold and $S_0,S_1\subset Z$ two subsets. A homeomorphism of pairs $g:(Z,S_0)\overset{\sim}{\to}(Z,S_1)$ (i.e. a homeomorphism $g:Z\overset{\sim}{\to}Z$ such that $g|_{S_0}:S_0\overset{\sim}{\to}S_1$ is also a homeomorphism) is called an \textit{ambient isotopy} from $S_0$ to $S_1$ in $Z$ if the following holds: There exists a family of subspaces $\{S_s\}_{s\in[0,1]}$ and a continuous family $\{g_s\}_{s\in[0,1]}$ of homeomorphisms $g_s:(Z,S_0)\overset{\sim}{\to}(Z,S_s)$ such that $g_0=\text{id}_Z$ and $g_1=g$. The map $[0,1]\times Z\to Z,\;(s,z)\mapsto g_s(z)$ is called a \textit{realization} of $g$.\footnote{This terminology varies slightly from what is found in \cite{pham}, where the family $\{S_s\}_{s\in[0,1]}$ of interpolating subspaces is called the realization.}
\end{defn}
We can think of an ambient isotopy as a continuous family of homeomorphisms deforming $S_0$ to $S_1$ within the ambient space $Z$, interpolating between $S_0$ and $S_1$ via a family of subspaces. To understand how a path in the parameter manifold $T$ generates an ambient isotopy we also need
\begin{defn}[\cite{pham}]\label{def:inverse_image_bundle}
    Let $\pi:Y\to T$ be a fiber bundle, $f:T'\to T$ a continuous map. Then the \textit{pull-back bundle}\footnote{Some authors in this field, like \cite{pham}, call this the inverse image bundle. Another common name is the \textit{induced bundle}.} of $\pi:Y\to T$ under $f$ is the fiber bundle $\pi':f^{-1} Y\to T'$ given by
    \begin{align*}
        f^{-1} Y:=\{(y,t')\in Y\times T' \;|\; \pi(y)=f(t')\}
    \end{align*}
    and $\pi'(y,t'):=t'$. This definition naturally extends to bundles of pairs. For a fiber bundle of pairs $\pi:(Y,S)\to T$ we denote the corresponding pull-back bundle of pairs by $\pi':f^{-1}(Y,S)\to T'$.
\end{defn}
Recall that homotopic maps $f,g:T'\to T$ induce isomorphic pull-back bundles for sufficiently nice topological spaces $T$ and $T'$ (for example take $T$ and $T'$ to be CW-complexes). As an easy consequence we obtain that fiber bundles over a contractible CW-complex are (globally!) trivial. The same is true for locally trivial fiber bundles of pairs. Hence if we are given a locally trivial bundle of pairs $\pi:(Z\times T,\mathcal{S})\to T$ and a path $\gamma:[0,1]\to T$, the pull-back bundle $\pi':(Z\times[0,1],\mathcal{S}')\to[0,1]$ is trivial as a bundle of pairs. So we conclude that $\gamma$ defines an ambient isotopy from $\mathcal{S}'\cap(\{0\}\times Z)$ to $\mathcal{S}'\cap(\{1\}\times Z)$.\\
Now we return to the the original question of analyticity of integrals. Suppose
\begin{equation}\label{eq:standard_form}
        \int_{\Gamma_0}\frac{\omega(t_0)}{s(t_0)}
\end{equation}
is an integral of standard form. First we note that, since $\omega(t_0)$ is closed as well as analytic and $s(t_0)$ is analytic on $Z-S(t_0)$, the form $\frac{\omega(t_0)}{s(t_0)}$ is a closed analytic $n$-form on $Z-S(t_0)$. Hence the result of the integration depends only on the homology class of $\Gamma_0$ by Stoke's theorem. We may differentiate with respect to $t$ at $t_0$ under the integral sign since $\Gamma$ is compact, $S(t)$ closed, and $\frac{\omega(t)}{S(t)}$ is analytic at $t_0$. So equation (\ref{eq:standard_form}) defines a germ of an analytic function at $t_0$. Denote $\mathcal{S}:=\bigcup_{t\in T}S(t)\times\{t\}\subset Z\times T$. If now $\gamma$ is a path $\gamma:[0,1]\to T$ from $t_0$ to some $t\in T$ such that the pull-back bundle of $\pi:(Z\times T,\mathcal{S})\to T$ by $\gamma$ is locally trivial and $\sigma:[0,1]\times Z\to Z$ is a realization of the induced ambient isotopy we can analytically continue along $\gamma$ via
\begin{equation}
        \int_{\sigma(s,\Gamma_0)}\frac{\omega(\gamma(s))}{s(\gamma(s))}.
\end{equation}
In particular we have the following theorem:
\begin{thm}[\cite{app-iso}]\label{thm:standard_form}
    Let
    \begin{align*}
        \int_{\Gamma_0}\frac{\omega(t_0)}{s(t_0)}
    \end{align*}
    be an integral of standard form such that for $\mathcal{S}:=\bigcup_{t\in T}S(t)$ the fiber bundle of pairs $\pi:(Z\times T,\mathcal{S})\to T$ is locally trivial. Then the integral defines a germ of an analytic function at $t_0$ which can be analytically continued along any path in $T$. Moreover the continuation at $t_1\in T$ is of the form
    \begin{align*}
        \int_{\Gamma}\frac{\omega(t_1)}{s(t_1)}
    \end{align*}
    with $\Gamma\in H_n^c(Z-S(t_1))$ (where $H_n^c(X)$ denotes the $n$th compact homology group of a topological space $X$).
\end{thm}
Note that, since local triviality is a local condition, Theorem \ref{thm:standard_form} implies that an integral of standard form always defines an analytic function on the path-connected component of
\begin{align*}
    \{t\in T \;|\; \pi:(Z\times T,\mathcal{S})\to T\text{ is locally trivial at }t\}\subset T
\end{align*}
containing $t_0$.

\section{Quadratic Integrals}\label{sec:quadric_integrals}
In this section we introduce the class of integrals under investigation in this paper. To do so we first need to set up some notation: We denote the space of $m\times n$-matrices with coefficients in a ring $R$ by $M(m\times n;R)$. For any matrix $M\in M(m\times n;R)$ we denote its transpose by $M^T\in M(n\times m;R)$. For a vector $z\in R^n$ we denote $z^2:=z_1^2+\cdots+z_n^2$ where $z_i$ is the $i$th component of $z$. Furthermore we sometimes write $z_1z_2:=z_1^Tz_2$ for $z_1,z_2\in R^n$ if this does not lead to confusion. Let $z\in R^n\simeq M(n\times1,R)$ and $M\in M(n\times n;R)$. We always identify the $1\times1$ matrix $z^TMz$ with its single entry in $\mathbb{C}$ without explicit mention. Furthermore we say (by slight abuse of language) that $M$ is positive definite (resp. positive semi-definite) if all its entries are real and if it is positive definite (resp. positive semi-definite) as a matrix in $M(n\times n;\mathbb{R})\subset M(n\times n;\mathbb{C})$. The adjugate matrix of $M$ is denoted by $\text{adj}\,M$. Let $U\subset\mathbb{C}^m$ be an open set. By a family of quadratic functions $\{Q(t)\}_{t\in U}$ on $\mathbb{C}^n$ analytically dependent on $t\in U$ we mean a family of functions
\begin{align*}
    Q(t):\mathbb{C}^n\to\mathbb{C},\quad z\mapsto z^TM(t)z+2a(t)^Tz+b(t)
\end{align*}
where $M(t)\in M(n\times n;\mathbb{C})$ is a symmetric matrix, $a(t)\in\mathbb{C}^n$, and $b(t)\in\mathbb{C}$ all analytic in $t$. For $a\equiv b\equiv 0$ this is a quadratic form and we say that $\{M(t)\}_{t\in U}$ is the family of matrices representing $\{Q(t)\}_{t\in U}$. Of course this only makes sense with respect to a basis which we fix once and for all to be, say, the canonical one (if not explicitly mentioned otherwise) for all $t\in U$. We can turn such a family of quadratic functions into a family of quadratic forms by adding one auxiliary variable, homogenizing the expression. For a quadratic function $Q:\mathbb{C}^n\to\mathbb{C}$ with $Q(z)=z^TMz+2a^Tz+b$ we denote this by
\begin{align*}
    H(Q):\mathbb{C}^{n+1}\to\mathbb{C},\quad (z_0,z)&\mapsto z^TMz+2z_0a^Tz+z_0^2b\\
    &\;=\begin{pmatrix}z_0\\ z\end{pmatrix}^T\tilde{M}\begin{pmatrix}z_0\\ z\end{pmatrix},
\end{align*}
where
\begin{align*}
    \tilde{M}:=\begin{pmatrix} b & a^T \\ a & M \end{pmatrix}\in M((n+1)\times(n+1);\mathbb{C}).
\end{align*}
Note that if $\{Q(t)\}_{t\in U}$ is an analytic family of quadratic functions then $\{H(Q(t))\}_{t\in U}$ still depends analytically on $t$. More generally for any finite linear combination of power functions $p(z)=\sum_\alpha c_\alpha z^\alpha$ (where we employ the usual multi-index notation and the sum is over a finite set of complex numbers $\alpha$) of degree $\text{deg}\,p\in\mathbb{C}$ (which we define to be the exponent with the largest real part) we define the corresponding homogeneous function $H(p)$ by
\begin{align*}
    H(p)(z_0,z):=\sum_\alpha c_\alpha z_0^{\deg\,p-\sum_{i=1}^n\alpha_i}z^\alpha.
\end{align*}
The integrals (\ref{eq:introintegral}) from the introduction are not of standard form. To have a chance of applying the isotopy techniques to such an integral we first need to compactify the integration domain and the ambient space. Such a compactification is of course not unique and we have to make a choice.
\begin{defn}\label{def:quadric_integral}
    Let $U\subset \mathbb{C}^m$ be open, $\{Q_1(t)\}_{t\in U},\ldots,\{Q_N(t)\}_{t\in U}$ be analytic families of quadratic functions on $\mathbb{C}^n$, $\{p(t)\}_{t\in U}$ an analytic family of finite linear combinations of power functions on $\mathbb{C}^n$ with exponents having non-negative real part and let $\lambda_1,\ldots,\lambda_N\in\mathbb{C}$ with $\text{Re}\lambda_i>0$ for all $i\in\{1,\ldots,N\}$. Denote $\lambda:=\sum_{i=1}^N\lambda_i$. We call an integral of the form
    \begin{align*}
        I(t)=\int_{\mathbb{R}^n}\frac{p(t)(y)}{\prod_{i=1}^N(Q_i(t)(y))^{\lambda_i}}dy_1\wedge\cdots\wedge dy_n
    \end{align*}
    a \textit{quadratic integral} on $U$.\\
    If moreover the $Q_i(t)$ are quadratic forms for all $i\in\{1,\ldots, N\}$ and $t\in U$, $n$ is odd and $p(t)$ is of degree $2\lambda-n-1$ for all $t\in U$, we call an integral of the form
    \begin{align*}
        I(t)=\int_{\mathbb{R}\mathbb{P}^n}\frac{p(t)(z)\cdot\Omega_n}{\prod_{i=1}^N(Q_i(t)(z))^{\lambda_i}}
    \end{align*}
    a \textit{projective quadratic integral} on $U$ where $\Omega_n:=\sum_{i=0}^n(-1)^iz_idz_0\wedge\cdots\wedge\widehat{dz_i}\wedge\cdots\wedge dz_n$.
\end{defn}
The integrals in this definition should a priori be read as formal expressions, i.e. as just the collection of the set $U$, the analytic families of quadratic functions or forms $\{Q_i(t)\}_{t\in U}$, the analytic family $\{p(t)\}_{t\in U}$, the exponents $\lambda_1,\ldots,\lambda_N$, and the natural numbers $m$ and $n$. As long as no confusion can arise we always abbreviate by simply writing $I$ for a (projective) quadratic integral. If for a given integral $I$ on $U$ there is a subset $V\subset U$ such that $I(t)$ converges for all $t\in V$ then $I$ defines a map $V\to\mathbb{C}$.\footnote{To be a bit pathological: Since we technically allow $V$ to be empty this condition is of course always fulfilled. This will cause no confusion in what follows.} If $I(t)$ converges for a given $t\in T$ or not is of course heavily dependent on the values of $\lambda_1,\ldots,\lambda_N$. We always implicitly identify the formal integral with the corresponding function when appropriate without mentioning so explicitly.
Non-projective and projective quadratic integrals are related by the following simple proposition.
\begin{prop}\label{thm:projectivize}
    Let $n\in\mathbb{N}$ be odd. Then the identity
    \begin{align*}
    \int_{\mathbb{R}^n}\frac{p(t)(y)dy_1\wedge\cdots\wedge dy_n}{\prod_{i=1}^N(Q_i(t)(y))^{\lambda_i}}=
    \int_{\mathbb{R}\mathbb{P}^n}\frac{H(p(t))(z)z_0^{2\lambda-n-1-\text{deg}\,p(t)}}{\prod_{i=1}^N(H(Q_i(t))(z))^{\lambda_i}}\cdot\Omega_n
    \end{align*}
    holds whenever one (and hence both) of the two integrals converges. In particular every quadratic integral with $2\textup{Re}(\lambda)-n-1-\textup{deg}(p)\geq0$ or $2\textup{Re}(\lambda)-n-1-\textup{deg}(p)=-2$ can be identified with a projective one.
\end{prop}
\begin{proof}
    We consider the inclusion
    \begin{align*}
        i:\mathbb{C}^n\hookrightarrow\mathbb{C}\mathbb{P}^n,\quad (x_1,\ldots,x_n)\mapsto [1:x_1:\cdots:x_n].
    \end{align*}
    It restricts to a diffeomorphism onto its image. We transform the quadratic integral via the pull-back of the inverse of this restriction. After cancellation of factors of $z_0$, this pull-back of the integrand can readily be extended to the form given above. Note that we added the real hyperplane at infinity to the integration domain. This is a set of measure zero and does not change the result of the integration.
\end{proof}
Here we needed to restrict ourselves to odd $n$ due to orientability. This is however not a serious restriction since we can augment a quadratic integral by an auxiliary integration $1=\int_\mathbb{R}\frac{dx}{x^2+\pi^2}$ to get from an even $n$ to an odd one. Adding such a dummy integration has no influence on our arguments as we will see further below. In the rest of this paper we drop this dummy integration from our notation. We are interested in certain regular behavior of the (projective) quadratic functions involved in these integrals which will allow us to apply the isotopy techniques from Subsection \ref{sec:isotopy}. Hence the following definition:
\begin{defn}
    Let $I$ be a projective quadratic integral on $U\subset\mathbb{C}^m$ and $M_1(t),\ldots,M_N(t)$ the matrices representing the $Q_1(t),\ldots,Q_N(t)$. We say that $I$ is \textit{regular at $t_0\in U$} if all matrices $M_1(t_0),\ldots,M_N(t_0)$ are positive definite, \textit{quasi-regular at $t_0\in U$} if they are positive semi-definite. If $V\subset U$ is a subset we say the integral $I$ is \textit{quasi-regular on $V$} if the matrices $M_1(t),\ldots,M_N(t)$ are positive semi-definite for all $t\in V$ and \textit{regular on $V$} if the matrices are positive definite for all $t\in V$.\\
    In both the projective and non-projective case we say that $I(t)$ is \textit{superficially convergent}\footnote{This term stems from the mathematics of Feynman integrals in quantum field theory. Note that $I(t)$ being superficially convergent does not imply that the integral $I(t)$ converges in the usual sense. It should also be noted that this does not quite agree with the usual meaning of superficial convergence in the case of Feynman integrals. To avoid confusion we emphasis this in Section \ref{sec:feynman_integral}, where we initially explicitly distinguish between the notion of being superficially convergent as a projective quadratic integral and simply being superficially convergent.} if $2\text{Re}\lambda\geq n+1+\text{deg}\,p(t)$ for all $t\in U$.
\end{defn}
While the term \enquote{superficially convergent} is adopted from the physics literature, the motivation to introduce this terminology here is Proposition \ref{thm:projectivize} which tells us that we can understand superficially convergent quadratic integrals as projective quadratic integrals.\\ 
As it turns out the regular integrals are rather easy to handle from the isotopy perspective. We choose the term \enquote{quasi-regular} because if a projective quadratic integral is quasi-regular we can still salvage the situation and obtain a regular integral by applying a regularization procedure introduced in Subsection \ref{sec:regularization}. It should be remarked that including the auxiliary integration $\int_\mathbb{R}\frac{dx}{x^2+\pi^2}$ in case of even $n$ generally turns an integral which is regular at some $t_0$ into one which is only quasi-regular at this point. The denominator of the dummy integration also needs to be regularized. A projective quadratic integral $I$ has a compact $n$-cycle as integration domain and the zero loci $V(Q_i(t))$ are closed analytic sets depending analytically on $t\in\mathbb{C}^m$ for all $i\in\{1,\ldots,N\}$. Moreover the integrand constitutes an analytic family of analytic $n$-forms on $\mathbb{C}\mathbb{P}^n$ (which we will take to be the compact ambient space) which are regular outside of an analytic set. Thus in this case $I$ satisfies almost all conditions to be of standard form. To understand $I$ as a function of $t$ we want to apply Corollary \ref{thm:main_corollary} which is of course not always possible since there might not be a $t_0\in U$ such that the integration domain and the zero loci do not intersect. Integrals which are regular at at least one point however have this desirable properties. We rely on the following fact: Let $U\subset\mathbb{C}^m$ and let $\{Q_1(t)\}_{t\in U},\ldots,\{Q_N(t)\}_{t\in U}$ be families of quadratic forms with corresponding matrix representatives $\{M_1(t)\}_{t\in U},\ldots,\{M_N\}_{t\in U}$. Suppose there exists $t_0\in U$ such that $M_i(t_0)$ has full rank for all $i\in\{1,\ldots,N\}$. Then the $M_1(t),\ldots,M_N(t)$ already (simultaneously) have full rank on an open dense subset of $U$ (see Proposition \ref{thm:basics_analytic_sets}) since the set of all $t\in U$ where the rank of one of the $M_i(t)$ is not maximal is determined by the vanishing of the product $\prod_{j=1}^N\det M_j(t)$ which is an analytic function in $t$ (for example by Leibniz formula for the determinant).
\begin{prop}\label{thm:regular_integral}
    Let $I$ be a projective quadratic integral on $U\subset\mathbb{C}^m$. If there exists $t_0\in U$ such that the integral is regular at $t_0$, then $I$ is of standard form. Moreover, if $S\subset U$ denotes the set of all $t$ such that the $V(Q_i(t))$ are not in general position on $\mathbb{C}\mathbb{P}^n$, there exists an open subset $V\subset U$ with $U-V\subset S$ such that for all $t\in V$ and all $i\in\{1,\ldots,N\}$ the zero locus $V(Q_i(t))$ is a complex analytic manifold of (complex) codimension 1.
\end{prop}
\begin{proof}
    For the first claim it suffices to show that $\mathbb{R}\mathbb{P}^n\cap S(t_0)=\emptyset$. Since the integral is assumed to be regular at $t_0$ the zero locus of each $Q_i(t_0)$ restricted to the reals is trivial. Thus $[z]\in\mathbb{R}\mathbb{P}^n$ implies $Q_i(z)\neq0$ (independent of the representative $z$) for all $1\in\{1,\ldots,N\}$.\\
    Due to regularity, there exists an open and dense subset $V\subset U$ by the remark above such that the matrices $M_i(t)$ representing the $Q_i(t)$ have full rank for all $t\in V$. Thus for $t\in V$ the differential of $Q_i(t)$ vanishes nowhere on $\mathbb{C}\mathbb{P}^n$ for all $1\in\{1,\ldots,N\}$. Hence by the implicit function theorem $Q_i(t)^{-1}(\{0\})$ is indeed a complex analytic manifold of codimension 1. Note that we indeed have $U-V\subset S$ since if $t\in U-V$ there is $z\in\ker\,M_i(t)$ with $z\neq0$. Since this is equivalent to $\frac{\partial}{\partial z}Q_i(t)(z)=0$, the zero loci are not in general position at $z$, i.e. $t\in S$.
\end{proof}
We conclude that a projective quadratic integral $I$ on $U\subset\mathbb{C}^m$ with a regular point can be viewed as an integral of standard form on the non-empty open set $\{t\in U \;|\; \forall i\in\{1,\ldots,N\}: \text{rank}\,M_i(t)=n\}$. All points $t\in U$ which we sacrifice by this restriction are classified as potentially non-analytic by Corollary \ref{thm:main_corollary} anyway. Hence we do not lose information this way.

\subsection{The Landau Surface}
By the previous Proposition \ref{thm:regular_integral} we now get explicit equations characterizing the potentially non-analytic points of $I$. For $U\subset\mathbb{C}^n$ open and $f\in\mathcal{O}(U)$ we denote $\frac{\partial}{\partial z}f:=(\frac{\partial}{\partial z_1}f,\ldots,\frac{\partial}{\partial z_n}f)$. We immediately obtain the following corollary:
\begin{cor}\label{thm:general_landau}
    Let $I$ be a projective quadratic integral and suppose there exists $t_0\in\mathbb{C}^m$ such that $I$ is regular at $t_0$. Then $I$ defines a germ of an analytic function at $t_0$ which can be analytically continued along any path in $\mathbb{C}^m\backslash L$, where $L$ is defined as the set of all $t\in\mathbb{C}^m$ such that there exists a solution $(\alpha,z,t)\in(\mathbb{C}^N\backslash\{0\})\times\mathbb{C}^{n+1}\times\mathbb{C}^m$ to
    \begin{align}\label{eq:landau_equations}
    \begin{split}
        \alpha_i=0\quad\text{or} & \quad Q_i(t)(z)=0\qquad\forall i\in\{1,\ldots,N\}\\
        & \sum_{i=1}^N\alpha_i\frac{\partial}{\partial z} Q_i(t)(z)=0.
    \end{split}
    \end{align}
    In particular $I$ defines a multivalued function on $\mathbb{C}^m\backslash L$.
\end{cor}
\begin{proof}
    Since $I$ is of standard form according to Proposition \ref{thm:regular_integral} we can directly apply Corollary \ref{thm:main_corollary}: The equations \ref{eq:landau_equations} are exactly the criterion for the zero loci $V(Q_i(t))$ to not be in general position. The first $N$ equations ensure that we only consider linear relations among gradients evaluated at points in the intersection of zero loci, the last equations simply says that these gradients are linearly dependent over $\mathbb{C}$. Note that by our definition of general position we actually have linear dependence over $\mathbb{R}$ of the gradients of $\text{Re}(Q_i(t))$ and $\text{Im}(Q_i(t))$ thought of as vectors in $\mathbb{R}^{2n}$. This is however equivalent to the second line of (\ref{eq:landau_equations}).
\end{proof}
Note that $L(I)$ is defined for any projective quadratic integral $I$, even if $I$ does not converge for any $t$ or does not satisfy any regularity-requirement. But in these cases it might coincide with the entire parameter space $\mathbb{C}^m$. In physics equations of this form are well-known in the case of Feynman integrals, introduced in the late 50s by Landau \cite{landau}. But physicists are usually interested in the solutions to (\ref{eq:landau_equations}) with more restrictions on the values of $\alpha$ and $z$. The reason for this is twofold: For one thing the primary focus is on specific values of the parameter $t$ in an $\mathbb
{R}$-vector subspace of $\mathbb{C}^m$ (see Section \ref{sec:feynman_integral}) and for another thing they are mainly interested in the evaluation of $I$ on just the principal branch. How this branch, often referred to as the \textit{physical branch}, is chosen exactly is often not stated explicitly. When we refer to the principal branch here, we only assume that it contains the germ at $t_0$ defined by the integral at hand.
\begin{defn}
    For any projective quadratic integral $I$ we denote the set of solutions to (\ref{eq:landau_equations}) by $c(I)$. We call $L(I):=\pi(c(I))$ the \textit{Landau surface} of $I$, where $\pi:(\mathbb{C}^N\backslash\{0\})\times\mathbb{C}^m\times\mathbb{C}^{n+1}\twoheadrightarrow\mathbb{C}^m$ is the canonical projection to the parameter space.
\end{defn}
In the literature regarding Feynman integrals this is often called the \textit{Landau variety}. In the context of a generic projective quadratic integral this name does not apply since, as the dependence in $t$ is assumed to be analytic but not necessarily polynomial, $L(I)$ need not be a variety. We do however know that $L(I)$ is an analytic set by the Proper Mapping Theorem \ref{thm:proper_mapping_thm} since $c(I)$ is an analytic set. In the context of physics, the dependence on $t$ is in fact polynomial so that $c(I)$ is really an algebraic set. We remark in passing that the word variety in the physical context is not always intended to imply that $L(I)$ is irreducible. If we allow for some of the parameters occurring in Feynman integrals to be fixed (this is a choice of the parameter manifold $T$ often used in physics) $L(I)$ is generally not irreducible. So the word variety in this context is often used synonymously with algebraic set.\\
For a generic projective quadratic integral $I$ it does not really make sense to single out a particular chart of $\mathbb{C}\mathbb{P}^n$. If however $I$ stems from a (non-projective) quadratic integral in the sense of Proposition \ref{thm:projectivize} the variable $z_0$ plays a particular role and it is sensible to decompose
\begin{align*}
    L(I)=L_\text{fin}(I)\cup L_\infty(I),
\end{align*}
where
\begin{align*}
    L_\text{fin}(I):=\{t\in L(I) \;|\; \exists (\alpha,z,t)\in c(I):z_0\neq0\}\\
    L_\infty(I):=\{t\in L(I) \;|\; \exists (\alpha,z,t)\in c(I):z_0=0\},
\end{align*}
to distinguish between singularities that arise at points at finite distance in the original ambient space $\mathbb{C}^n$ and those at infinite distance (visible as explicit points after compactification). Similarly we write $c_\text{fin}(I)$ and $c_\infty(I)$ for the corresponding sets of solutions. Physicists call these singularities of the first and second type respectively (compare to \cite{s-matrix}). The two sets $L_\text{fin}(I)$ and $L_\infty(I)$ are not necessarily disjoint. If $t\in L_\text{fin}(I)\cap L_\infty(I)$ it is sometimes said that $t$ constitutes a \textit{mixed singularity} \cite{s-matrix}.\\
A few elementary remarks on Landau surfaces: If $I$ is a projective quadratic integral and $I'$ is a projective quadratic integral such that the set of families of quadratic forms defining $I'$ is contained in that of $I$ we clearly have $L(I')\subset L(I)$. Furthermore if $I_1$ and $I_2$ are projective quadratic integrals then the functions $I_1I_2$ and $I_1+I_2$ are analytic outside $L(I_1)\cup L(I_2)$.\\
Of course the more interesting (projective) quadratic integrals are not regular. Remarking again on our motivation from physics, even if we consider integrals over Euclidean (instead of Minkowski) internal momenta, regularity is usually not satisfied at more than one loop. We show that equations like (\ref{eq:landau_equations}) still hold under certain circumstances including the case of Feynman integrals. In the next Subsection \ref{sec:regularization} we introduce a regularization method that allows us to apply Proposition \ref{thm:regular_integral} directly.

\subsection{Regularization of Quadrics}\label{sec:regularization}
As seen in the previous subsection, it would be advantageous to be able to work with quadratic integrals with a regular point. To this end we introduce a regularization which uses an analytically varying family of vectors which almost everywhere form a basis diagonalizing (in the usual sense of diagonalization of a quadratic form) the matrices. We construct a quadratic form which we add to the initial form with a small complex prefactor. Of course if the integral already has a quasi-regular point (which is our main interest here) this can easily be achieved by simply adding any positive definite form to each quadratic form. Since it is not much harder to show that a regular point can be achieved by adding a more general quadratic form and it seems useful to have this freedom of choice in practice, we conduct the discussion in a slightly more general context. The existence of the family of vectors as described above is guarantied by the following two lemmas:
\begin{lem}\label{thm:key_lemma_1}
    Let $U\subset\mathbb{C}^m$ be open, $\{M(t)\}_{t\in U}$ a family of complex symmetric $n\times n$-matrices analytically dependent on $t$. Then there exist a closed analytic subset $V$ with $\text{codim }V\geq1$ and analytic maps
    \begin{align*}
        v_1,\ldots,v_n:U\to\mathbb{C}^n
    \end{align*}
    such that for all $t\in U-V$ the vectors $v_1(t),\ldots,v_n(t)$ form a basis of $\mathbb{C}^n$ and $M(t)$ is diagonal with respect to this basis.\\
    Moreover the diagonal elements are analytic in $t$ and can be analytically continued to all of $U$. If $t_0\in U-V$ such that $M(t_0)$ is real-valued, we can choose the $v_i$ such that the $v_1(t_0),\ldots,v_n(t_0)$ are real.
\end{lem}
The proof of this lemma is basically just the diagonalization procedure for quadratic forms from linear algebra (which can be found in any good undergraduate book on linear algebra, for example in \cite{fischer}) while paying close attention to the analyticity with respect to $t$ of the families of vectors involved.
\begin{proof}
    We proceed by induction on $n$. If $n=1$ there is nothing to do, the matrices are already in diagonal form with respect to any fixed basis. Now for the induction step: If $M(t)=0$ for all $t\in U$ there is again nothing to prove (simply choose $v_1,\ldots,v_n$ as any basis constant in $t$). Otherwise let $v_1:U\to\mathbb{C}^n$ be an analytic map such that there exists $t_0\in U$ with $Q(t_0)(v_1(t_0))\neq 0$. Then
    \begin{align*}
        V:=\{t\in U \;|\; Q(t)(v_1(t))=0\}
    \end{align*}
    is a closed analytic set with $\text{codim }V\geq1$. Without loss of generality we may assume $v_1(t_0)\in\mathbb{R}^n$ if $M(t_0)$ is real valued. We define $a_j(t):=\sum_{i=1}^n(v_1(t))_i(M(t))_{ij}$. Note that not all of the $a_j(t)$ can be identically zero since that would imply $Q(t)(v_1(t))=0$ for all $t\in U$. For convenient notation we assume without loss of generality that $a_1(t)$ is not identically zero. Then the equation $v_1(t)^TM(t)x=0$ is solved with respect to $x$ by the analytic vector valued maps
    \begin{align*}
        b_2(t):=\begin{pmatrix}-a_2(t) \\ a_1(t) \\ 0 \\ \vdots \\ 0\end{pmatrix},\quad\ldots\quad,b_n(t):=\begin{pmatrix}-a_n(t) \\ 0 \\ \vdots \\ 0 \\ a_1(t)\end{pmatrix}
    \end{align*}
    which form as basis of the vector space $\{x\in\mathbb{C}^n \;|\; v_1(t)^TM(t)x=0\}$ outside of the closed analytic set
    \begin{align*}
        V':=\{t\in U \;|\; a_1(t)=0\}
    \end{align*}
    again with $\text{codim }V'\geq1$. Note that if $v_1(t_0)$ and $M(t_0)$ are real so are the $b_2(t_0),\ldots,b_n(t_0)$. With respect to the basis $v_1(t),b_2(t),\ldots,b_n(t)$ the matrix $M(t)$ takes the form
    \begin{align*}
        \begin{pmatrix}
            \lambda_1(t) & 0 & \cdots & 0 \\ 0 \\ \vdots & & M'(t) \\ 0
        \end{pmatrix}
    \end{align*}
    on $U-V\cup V'$ where $\lambda_1$ and $M'$ are analytic in $t$.
    Now $M'$ is of dimension $n-1$ and thus we can apply the induction hypothesis to $M'$ on the open set $U-V\cup V'$.\\
    We immediately see that $\lambda_1(t)=Q(t)(v_1(t))$ defines an analytic function which is bounded along $V\cup V'$ so that it can be analytically continued to $U$ due to Riemann's Extension Theorem \ref{thm:riemann_extension}.
\end{proof}
\begin{lem}\label{thm:key_lemma_2}
    Let $U\subset\mathbb{C}^m$ be open. If $\{M(t)\}_{t\in U}$ is a family of complex $n\times n$ matrices depending analytically on $t\in U$, there exist a closed analytic set $V$ with $\text{codim }V\geq1$ and a family $\{A(t)\}_{t\in U}$ of matrices depending analytically on $t$ such that $M(t)+\epsilon A(t)$ has full rank for all $\epsilon\in\mathbb{C}^\times$ and all $t\in U-V$.\\
    Moreover if there is $t_0\in U-V$ such that $M(t_0)$ is positive semi-definite, we can choose $\{A(t)\}_{t\in U}$ such that $A(t_0)$ is positive semi-definite as well and $M(t_0)+\epsilon A(t_0)$ is positive definite for all $\epsilon\in\mathbb{R}^+$.
\end{lem}
\begin{proof}
    Let $v_1,\ldots,v_n:U\to\mathbb{C}^n$ be as in Lemma \ref{thm:key_lemma_1}. Then outside of some closed analytic set $V$ with $\text{codim }V\geq1$ the $v_1(t),\ldots,v_n(t)$ form a basis such that $M(t)$ takes the form
    \begin{align*}
        \begin{pmatrix} \lambda_1(t) & & \\ & \ddots & \\ & & \lambda_n(t)\end{pmatrix}.
    \end{align*}
    with respect to this basis. By permuting the basis vectors we can assume without loss of generality that $\lambda_1,\ldots,\lambda_k$ ($k\leq n$) are the $\lambda_i$ which are not identically zero. Then $M(t)$ has rank $k$ outside of an analytic subset $V'$ with codimension greater or equal to 1. We define the family of matrices $A(t)$ by
    \begin{align*}
        A(t):=\text{adj}\,T(t)^T\begin{pmatrix} 0_{k\times k} & 0_{k\times n} \\ 0_{n\times k} & 1_n \end{pmatrix}\text{adj}\,T(t)=(\det T(t))^2\cdot (T(t)^{-1})^T\begin{pmatrix} 0_{k\times k} & 0_{k\times n} \\ 0_{n\times k} & 1_n\end{pmatrix}T(t)^{-1}
    \end{align*}
    where the second equality holds for all $t\in U-V$ and
    \begin{align*}
        T(t):=\begin{pmatrix}v_1(t)\cdots v_n(t)\end{pmatrix}\in M(n\times n;\mathbb{C})
    \end{align*}
    is the change of basis which is well-defined and analytic outside of $V$, $0_{m\times n}$ is the zero matrix of size $m\times n$, and $1_n$ the identity matrix of size $n\times n$. $A(t)$ can be analytically continued to all of $U$ since the adjoint matrix $\text{adj }T(t)$ contains as components products of the coefficients in $T(t)$ which are analytic. Now with respect to $v_1(t),\ldots,v_n(t)$ the matrix $M(t)+\epsilon A(t)$ takes the form
    \begin{align*}
        \begin{pmatrix} \lambda_1(t) & & & & & \\ & \ddots & & & & \\ & &  \lambda_k(t) & & & \\ & & & \epsilon(\det T(t))^2 & & \\ & & & & \ddots & \\ & & & & & & \epsilon(\det T(t))^2 \end{pmatrix}
    \end{align*}
    which has full rank for all $\epsilon\in\mathbb{C}^\times$ on $U-V\cup V'$. Now if $M(t_0)$ is positive semi-definite we may choose $v_1(t_0),\ldots,v_n(t_0)$ real so that we obtain a positive semi-definite $A(t_0)$ and a positive definite $M(t_0)+\epsilon A(t_0)$ for $\epsilon\in\mathbb{R}^+$.
\end{proof}
Using the existence of a family of matrices as described in the lemma we now define the regularization of quadratic forms and quadratic integrals (the terminology is borrowed from German where invertible matrices are also called \enquote{regulär}).\footnote{This notion is quite different from and not to be confused with the usual regularization (like dimensional or analytic regularization) of a Feynman integral. We review this concept in Subsection \ref{sec:feynman_regularization}.}
\begin{defn}
    Let $\{Q(t)\}_{t\in U}$ be an analytic family of quadratic forms. We define the corresponding \textit{$Q$-regularized} analytic family of quadratic forms $R_\epsilon(Q(t))$ by
    \begin{align*}
        R_\epsilon(Q(t))(z):=z^T(M(t)+\epsilon A(t))z
    \end{align*}
    where $\{A(t)\}_{t\in U}$ is an analytic family of matrices as in Lemma \ref{thm:key_lemma_2} and $\epsilon\in\mathbb{C}^\times$.
\end{defn}
Of course this definition (as well as the following one) depends on the choice of $\{A(t)\}_{t\in U}$. We nevertheless suppress this dependence in our notation for simplicity since none of what follows depends on the concrete choice. The only thing we always assume is that, if a quasi-regular point $t_0$ is specified, the family of matrices $\{A(t)\}_{t\in U}$ used for the regularization is such that for all $i\in\{1,\ldots,N\}$ the matrix $A_i(t_0)$ is positive semi-definite and $M_i(t_0)+\epsilon A_i(t_0)$ positive definite for all $\epsilon\in\mathbb{R}^+$.
\begin{defn}
    Let $I$ be a superficially convergent projective quadratic integral on $U\subset\mathbb{C}^m$ with families of quadratic forms $\{Q_1(t)\}_{t\in U},\ldots,\{Q_N(t)\}_{t\in U}$. We define the corresponding \textit{$Q$-regularized projective quadratic integral} by
    \begin{align*}
        R_\epsilon(I):=\int_{\mathbb{R}\mathbb{P}^n}\frac{p(z)\cdot\Omega_n}{\prod_{i=1}^N(R_\epsilon(Q_i(t))(z))^{\lambda_i}}.
    \end{align*}
\end{defn}
Note that a $Q$-regularized projective quadratic integral need not be regular. At this stage the regularization only ensures that the zero loci $V(Q_i(t))$ are (almost everywhere) smooth manifolds but in general there might not be a parameter $t$ making the zero loci and $\mathbb{R}\mathbb{P}^n$ disjoint. But as long as we have a quasi-regular point we indeed end up with a regular integral.
\begin{prop}\label{thm:regintstandardform}
    Let $I$ be a projective quadratic integral on $U\subset\mathbb{C}^m$ which is quasi-regular at $t_0\in U$. Then there is a $Q$-regularization such that the corresponding $Q$-regularized integral $R_\epsilon(I)$ is regular at $t_0$ for all $\epsilon\in\mathbb{R}^+$. In particular it is of standard form on $U$ and the zero loci $V(R_\epsilon(Q_i(t)))$ are complex analytic manifolds of codimension 1 for all $t\in U\backslash L(I)$.
\end{prop}
\begin{proof}
    By Proposition \ref{thm:regular_integral} it suffices to show that there is a $Q$-regularization such that $R_\epsilon(I)$ is regular at $t_0$. The existence of such a $Q$-regularization is precisely the content of Lemma \ref{thm:key_lemma_2}.
\end{proof}
Of course we want to obtain information about our original integral. Clearly we have $\lim_{\epsilon\to0}R_\epsilon(Q)=Q$ (in the sense of pointwise convergence) for any quadratic form $Q$ and any $Q$-regularization. Hence the integrand is recovered in the limit $\epsilon\to0$, at least wherever it is defined. By Corollary \ref{thm:general_landau} the set on which $R_\epsilon(I)$ is potentially non-analytic is defined by equations on which we can also take the limit $\epsilon\to0$. The content of the next theorem is that this yields indeed the set on which $\lim_{\epsilon\to0}R_\epsilon(I)$ is potentially not analytic.
\begin{thm}\label{thm:quadric_integrals_main_theorem}
    Let $I$ be a projective quadratic integral on $U\subset\mathbb{C}^m$, quasi-regular at $t_0\in U$. If $I$ converges as an integral on an open set containing $t_0$ then $I$ defines a multivalued function on $\mathbb{C}^m\backslash L(I)$.
\end{thm}
\begin{proof}
    Let $D\subset\mathbb{C}$ be the open unit disc around 0 and $R_\epsilon(I)$ a $Q$-regularization of $I$ as in Proposition \ref{thm:regintstandardform}. We set $L_\epsilon:=L(R_\epsilon(I))$ and denote by $\{\mathcal{A}_i(t)\}_{t\in U}$ the analytic family of quadratic forms associated to the family of matrices $\{A_i(t)\}_{t\in U}$ from the $Q$-regularization for all $i\in\{1,\ldots,N\}$. We can view the regularized integral $R_\epsilon(I)(t)$ as a projective quadratic integral $R_\bullet(I)$ on $T\times D$. We know that $R_\bullet(I)$ is of standard form but we can not directly apply Corollary \ref{thm:main_corollary} since the zero loci $V(Q_i(t)+\epsilon\mathcal{A}_i(t))$ might not be manifolds for $\epsilon=0$. For $\epsilon\neq0$ however the $Q$-regularization guarantees that this is the case so we apply Corollary \ref{thm:main_corollary} to the restriction $R_\bullet(I)|_A$ where $A:=T\times(D-\{0\})$. Hence $R_\bullet(I)|_{A}$ defines a germ of an analytic function at $(t_0,\epsilon_0)$, where $\epsilon_0\in\mathbb{R}^+$, which can be analytically continued to all $(t,\epsilon)\in A-\bigcup_{\epsilon\in D-\{0\}}(L_\epsilon\times\{\epsilon\})$. Now the question is if we can analytically continue to $(T\times D)-\bigcup_{\epsilon\in D}(L_\epsilon\times\{\epsilon\})$. Note that $(T\times D)-\bigcup_{\epsilon\in D}(L_\epsilon\times\{\epsilon\})$ is an open subset of $\mathbb{C}^m\times\mathbb{C}\simeq \mathbb{C}^{m+1}$ and that $L_0\times\{0\}$ is an analytic set (given by the Landau equations and the analytic equation $\epsilon=0$) of codimension $\geq2$ (since $I$ converges for some $t$, hence $L_0\neq\mathbb{C}^m$ and $L_0$ has codimension $\geq1$ according to Proposition \ref{thm:basics_analytic_sets}). Thus according to Riemann's Second Extension Theorem \ref{thm:riemann_extension_2} we can indeed analytically continue as desired.
\end{proof}

\subsection{Comment on the Use of Thom's Isotopy Theorem}
A very useful tool to determine if a fiber bundle of pairs is locally trivial is the Isotopy Theorem due to Thom \cite{thomiso}. We quickly recall the notions of a stratification and Whitney regularity needed to state the theorem. In this section a fiber bundle is always locally trivial (we emphasize this because in the case of fiber bundle of pairs, see Definition \ref{def:fibre_pair}, we use the convention to explicitly state if a fiber bundle of pairs is (locally) trivial or not.).
\begin{defn}[\cite{pham}]
    Let $X$ be a topological space. A \textit{stratification} of $X$ is a sequence of closed sets
    \begin{align*}
        X=S^{n_1}\supset S^{n_2}\supset\cdots\supset S^{n_k}\supset\cdots
    \end{align*}
    such that
    \begin{enumerate}
        \item $S^{n_ i}-S^{n_{i+1}}$ is a smooth manifold of dimension $n_i$. The connected components of the $S^{n_ i}-S^{n_{i+1}}$ are called the \textit{strata}.
        \item The boundary of any stratum is the union of strata of strictly smaller dimension.
    \end{enumerate}
\end{defn}
In \cite{whitney} Whitney introduced some additional conditions on the way two strata are allowed to be incident to each other. They refer to the convergence of tangent planes and lines which is to be understood as convergence in the relevant Grassmannians.
\begin{defn}[Whitney's Regularity Conditions, \cite{mather}]
    Let $(X,Y)$ be a pair of two smooth manifolds $X$ and $Y$ and let $y\in Y$.
    \begin{enumerate}
        \item We say $(X,Y)$ is \textit{Whitney a-regular} at $y$ if for any sequence $\{x_n\}_{n\in\mathbb{N}}$ in $X$ converging to $y$ such that the tangent planes $T_{x_n}X$ converge to some plane $\tau$ we have $TY_y\subset\tau$.
        \item We say $(X,Y)$ is \textit{Whitney b-regular} at $y$ if for any sequences $\{x_n\}_{n\in\mathbb{N}}$ in $X$ and $\{y_n\}_{n\in\mathbb{N}}$ in $Y$, both converging to $y$, such that the tangent planes $T_{x_n}X$ converge to some plane $\tau$ and the sequence of secants $\overline{x_ny_n}$ converges to a line $l$ then $l\subset\tau$.
    \end{enumerate}
    We say the pair $(X,Y)$ is Whitney a-regular (resp. Whitney b-regular) if it is Whitney a-regular (resp. Whitney b-regular) at every point in $Y$. A stratification of a topological space is called Whitney a- or b-regular if all pairs of strata are Whitney a- or b-regular.
\end{defn}
A Whitney b-regular stratification is simply called a \textit{Whitney stratification}. Although this is the form in which the Whitney conditions are commonly stated, it suffices to consider the second one since Mather proved that Whitney b-regularity implies Whitney a-regularity \cite{mather}.\footnote{Whitney a-regularity is still used in practice: There is a weaker requirement than Whitney b-regularity, sometimes called \textit{Whitney b'-regularity}, which together with a-regularity implies b-regularity \cite{lu}. This weaker condition is often easier to check.}
\begin{defn}[\cite{pham}]
    Let $Y$ and $T$ be differentiable manifolds with $T$ connected. If $\pi:Y\to T$ is a differentiable map we say that $\pi:(Y,S)\to T$ is a \textit{stratified bundle} if there is a stratified set $X$ such that $\pi:Y\to T$ is a (locally trivial) fiber bundle with fiber $X$ and the local homeomorphisms $\phi_U$ maps strata of $Y$ to the product of a stratum of $X$ with $T$.
\end{defn}
Given a stratified subset $S\subset Y$ we can obtain a stratification of $Y$ by taking the strata of $S$ together with the connected components of $Y\backslash S$. If the stratification of $S$ is Whitney b-regular, so is the resulting stratication of $Y$ \cite{pham}. Then if $\pi:Y\to T$ is a stratified bundle obtained in this manner, the fiber bundle of pairs $\pi:(Y,S)\to T$ is locally trivial. We can now state the theorem:
\begin{thm}[Thom's Isotopy Theorem, \cite{pham}]\label{thm:thoms_lemma}
    Let $Y$ and $T$ be smooth manifolds, $T$ connected and $Y$ equipped with a Whitney stratification. Let $\pi:Y\to T$ be a proper differentiable map. If the restriction of $\pi$ to each stratum is a submersion then $Y$ is a stratified bundle.
\end{thm}
This statement is sometimes employed in the particular context of Feynman integrals (see for example \cite{periods} or \cite{boyling}), but usually in the context of the parametric representation. In fact, often times the Landau variety is defined by means of this theorem. For example in \cite{periods} it is defined as the codimension 1 part of $\pi(\bigcup cA_i)$ where the union is over all strata $A_i$ and $cA_i$ is the set of all points where $\pi|_{A_i}$ fails to be a submersion. It is certainly a more powerful criterion than Corollary \ref{thm:main_corollary} to test for local triviality. But nevertheless the main obstacle for quadratic integrals (from the perspective of reducing the problem to the study of integrals of standard form) can not be overcome by Thom's Isotopy Theorem, namely that in general and in particular for most Feynman integrals there is no appropriate $t_0$ in the parameter space to start with. If there was such a $t_0$, then the matrix representatives of the $Q_i(t)$ would already have full rank almost everywhere so that Corollary \ref{thm:main_corollary} would apply anyway. Thus in the context considered here the theorem is not particularly useful.\\
While it is evident that the singular locus of the integrand admits the structure of a Whitney stratification,\footnote{In \cite{whitney} Whitney proved that any analytic variety admits a Whitney stratification. Several generalizations, for example to semi-algebraic sets \cite{thom} or subanalytic set \cite{hironaka}, have been established since then.} it does not seem immediately clear if the natural stratification
\begin{align}\label{eq:stratification}
    S_i-\bigcup_{j\neq i}S_i\cap S_j,\quad S_i\cap S_j-\bigcup_{k\neq i, k\neq j}S_i\cap S_j\cap S_k,\quad \ldots
\end{align}
with $S_i:=\{(z,t)\in\mathbb{C}\mathbb{P}^n\times U\;|\; H(Q_i)(t)(z)=0\}$ for all $1\leq i\leq N$ is Whitney regular (for some open $U\subset\mathbb{C}^m$ chosen to include the physically relevant points). If this were true we could directly deduce the Landau equations from Theorem \ref{thm:thoms_lemma}.

\section{Feynman Integrals}\label{sec:feynman_integral}
Perturbative quantum field theory is dominated by the study of Feynman integrals which are obtained from graphs with some additional structure by applying the so called \textit{Feynman rules}. The general philosophy of this procedure is that, roughly speaking, to obtain a prediction of a certain result's probability in the quantum world we need to \enquote{sum} the probability amplitudes of all experimentally indistinguishable ways this result can occur. A Feynman integral then represents one such contribution to the probability amplitude.\\
There are multiple possible definitions for such graphs with additional structure, generally called Feynman graphs, to capture the relevant combinatorial structure, each having their own advantages and disadvantages. For our purposes we define a Feynman graph as follows:
\begin{defn}
    A \textit{Feynman graph} is a pair $(G,\phi)$ of a finite undirected multigraph $G$, i.e. a pair $G=(V(G),E(G))$ of a finite set $V(G)$ whose elements are called \textit{vertices} of $G$ and a finite multiset $E(G)$ whose elements are unordered pairs of vertices called \textit{edges} of $G$, together with a map $\phi:V(G)\to\mathbb{N}$, the \textit{external structure}, assigning each vertex the number of external lines incident to it. We denote the set of all Feynman graphs by $\mathcal{G}$.
\end{defn}
The edges of a Feynman graph represent interacting (virtual) particles and the external structure encodes the particles which we consider incoming or outgoing. In a collider experiment for example the particles we prepare for a collision may count as incoming, the particles measured as a result of the collision as outgoing. Physicists often work with a definition based on half-edges instead. Then the half-edges not joined to a full edge represent incoming and outgoing particles. This replaces the map $\phi$ from our definition. For a definition of Feynman graphs using half-edges see for example \cite{coloredgraphs}. Aside from numerous equivalent variations of this definition there is also a number of generalizations. If one wants to consider different types of particles for example, this can be done by augmenting the definition above by a map $c:E\to\{1,...,N_p\}$ where $N_p$ is the number of particles allowed in the theory under consideration. In this work we restrict ourselves to scalar theories and we do not distinguish different types of edges.\footnote{Most obstructions preventing analyticity can already be observed in this setting. A detailed research is reserved for the future.} An example of a common pictorial representation of Feynman graphs is given in Figure \ref{fig:qed_diagrams}: It shows two graphs occurring in quantum electrodynamics (QED for short), both contributing to the scattering amplitude of an electron and a positron. In QED there are three types of edges: The wiggly lines represent photons while the straight lines with an arrow represent electrons or positrons depending on their orientation (with respect to the horizontal axis). The external structure is indicated by the lines which are not connected to a vertex on one end.\\
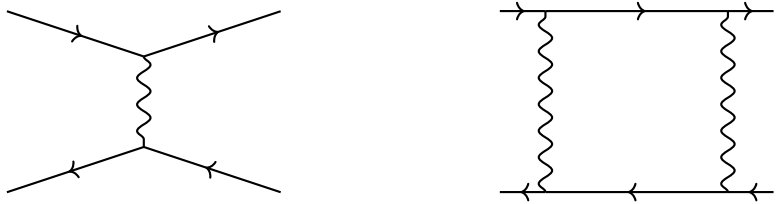
\begin{figure}
\begin{subfigure}[c]{0.4\textwidth}
\begin{tikzpicture}[thick, scale=0.6]
    \draw (0,2) edge [electron] (3,1);
    \draw (3,1) edge [electron] (6,2);
    \draw (3,-1) edge [electron] (0,-2);
    \draw (6,-2) edge [electron] (3,-1);
    \draw (3,1) edge [photon] (3,-1);
\end{tikzpicture}
\end{subfigure}
\begin{subfigure}[c]{0.4\textwidth}
\begin{tikzpicture}[thick, scale=0.6]
    \draw (-1,4) edge [electron] (0,4);
    \draw (0,4) edge [electron] (4,4);
    \draw (4,4) edge [electron] (5,4);
    \draw (4,0) edge [photon] (4,4);
    \draw (0,0) edge [photon] (0,4);
    \draw (0,0) edge [electron] (-1,0);
    \draw (4,0) edge [electron] (0,0);
    \draw (5,0) edge [electron] (4,0);
\end{tikzpicture}
\end{subfigure}
\caption{Two Feynman graphs in quantum electrodynamics.}
\label{fig:qed_diagrams}
\end{figure}
If the external structure is evident or of no specific relevance we denote  a Feynman graph $(G,\phi)$ simply by $G$. Let $(G,\phi)$ be a Feynman graph. For any vertex $v\in V(G)$ we denote its \textit{star}, i.e. the set of all edges incident to $v$, by $\text{st}(v)$. Furthermore we denote the number of independent cycles in $G$ by $h_1(G)$ (also referred to as the \textit{first Betti number} of $G$) and write $N_G:=\sum_{v\in V(G)}\phi(v)$ for the total number of external momenta. We say that $G$ is \textit{one-particle irreducible} (1PI for short) if $G$ is connected and still connected upon removing any edge $e\in E(G)$. In the mathematical literature this property is often called \textit{2-edge-connected}.\footnote{Unfortunately, the terminology in the physics and mathematics community often differs greatly in this area, as this sentence exemplifies. This makes it hard to write for an audience in the intersection and the author apologizes for any inconvenience this might cause here.} A \textit{Feynman subgraph} $(\gamma,\psi)$ of $(G,\phi)$, denoted by $\gamma\subset G$ (by abuse of notation), is a Feynman graph such that $V(\gamma)\subset V(G)$, $E(\gamma)\subset E(G)$, and $\psi=\phi|_{V(\gamma)}$. We denote the Feynman graph obtained from shrinking a subgraph $\gamma\subset G$ to a vertex by $G/\gamma$ (with the obvious external structure).\\
In perturbative quantum field theory, there is a family of (not necessarily convergent) integrals associated to a Feynman graph. We give a lightning review of the theory of Feynman integrals. For details on Feynman graphs and integrals in perturbative quantum field theory the reader can consult classical works like \cite{weinberg} or a more modern exposition like \cite{weinzierl}, focusing on the perturbative aspects of quantum field theory.\\
The procedure to assign an integral to a Feynman graph goes as follows: First we fix a \textit{space-time dimension} $D\in\mathbb{N}$ (in Subsection \ref{sec:feynman_regularization} we comment on dimensional regularization where $D$ is allowed to vary). Let $(G,\phi)$ be a Feynman graph. In practice one usually restricts attention to a certain family of Feynman graphs since a given theory does not allow for vertices of arbitrary valency (see again Figure \ref{fig:qed_diagrams} where only the vertex joining a photon, an electron, and a positron is allowed which is in particular a 3-valent vertex). Since this is of no consequence to our exposition (among other things since we only look at individual graphs until Section \ref{sec:renormalization}), we do not put any such restriction on our graphs. To each edge $e\in E(G)$ we associate a \textit{mass} $m_e\in\mathbb{R}_{\geq0}$ (or more generally $m_e\in\mathbb{C}$) and to each vertex $v\in V(G)$ a set of \textit{external momenta} $\{p_{v,i}\}_{1\leq i\leq\phi(v)}$ with $p_{v,i}\in\mathbb{C}^D$. The result of the integration depends only on $p_v:=\sum_{i=1}^{\phi(v)}p_{v,i}$. Fix a small $\epsilon>0$ and a vertex $v_0\in V(G)$. The integral is then obtained by writing a factor of
\begin{align}\label{eq:propagator}
    \frac{1}{(k_e^2+m_e^2-i\epsilon)^{\lambda_e}}
\end{align}
for each edge $e\in E(G)$ with $\lambda_e\in\mathbb{C}$ (called an \textit{analytic regulator}), then enforcing a linear condition $p_v+\sum_{e\in E}\mathcal{E}_{e,v}k_e=0$ at every vertex $v\in V(G)\backslash\{v_0\}$ and finally integrating over the $k_e$ not determined by these conditions. The reason for leaving out the condition for $v_0$ is to separate the overall momentum conservation $\sum_{v\in V(G)}p_v=0$ of the external momenta from the remaining conditions (note that we agree here on the convention that all external momenta are considered to be incoming). These remaining $k_e$ are referred to as \textit{loop} or \textit{internal momenta} since their number is equal to $h_1(G)$.\footnote{While graph theorists often call edges incident to only one vertex a loop, physicists commonly use the term loop for independent cycles.} The nominator of (\ref{eq:propagator}) now depends on the $p_v$ and we denote the resulting expression by $Q_e(p,m)(k)$ for all $e\in E(G)$ where $p=(p_1,\ldots,p_{|V(G)|})\in\mathbb{C}^{N_GD}$ and $m=(m_1,\ldots,m_{|E(G)|})\in\mathbb{C}^{|E(G)|}$. Thus we obtain an analytic family $\{Q_e(p,m)\}_{(p,m)\in\mathbb{C}^{N_GD}\times\mathbb{C}^{|E(G)|}}$ of quadratic functions on $\mathbb{C}^{h_1(G)D}$. In this prescription we actually need to temporarily endow $G$ with an orientation: $\mathcal{E}_{e,v}\in M(|E(G)|\times|V(G)|;\mathbb{Z})$ is the incidence matrix of $G$ whose definition is orientation dependent. The final result of the integration (whenever it makes sense) can be shown to be independent of the chosen orientation. In our treatment we work in Minkowski-space with signature $(-,+,\ldots,+)$\footnote{Particle physicists tend to dislike this convention (see for example \cite{fun}, footnote on page 2), working rather with $(+,-,\ldots,-)$, while people working in general relativity prefer the one used here. In our application we can avoid some factors of $-1$ by sticking to the second one.} which means we integrate over $i\mathbb{R}\times\mathbb{R}^{D-1}$ for each loop momentum. The general Feynman integral in a scalar theory thus takes the form of an integral
\begin{align}\label{eq:feynman_integral}
    I_G(p,m)=\left(\prod_{i=1}^{h_1(G)}\int_{\mathbb{R}^{D-1}}d^{D-1}\vec{k}_i\int_{-i\infty}^{i\infty}d(k_i)_0\right)\frac{1}{\prod_{e\in E}(Q_e(p,m)(k))^{\lambda_e}}.
\end{align}
We have omitted from equation (\ref{eq:feynman_integral}) some factors of $i$s and $\pi$s as well as coupling constants, which are needed to obtain a numerical agreement with the experiment but play no role in the analytical structure, to slightly declutter the notation.

\subsection{Wick Rotation}
Without the $i\epsilon$-term in (\ref{eq:propagator}) the integral (\ref{eq:feynman_integral}) would not be well-defined. There are several ways to think about this so called $i\epsilon$-prescription. One might for example say that we do not take $\epsilon$ to be some small finite number but consider instead
\begin{align*}
    \frac{1}{k_e^2+m_e^2-i0}=\lim_{\epsilon\to0^+}\frac{1}{k_e^2+m_e^2-i\epsilon}
\end{align*}
as the meaning of the propagator. Of course this limit does not exist in the space of smooth functions due to the pole at $k_e^2=-m_e^2$. Rather one is forced to view it as a distribution and one is lead down the road of Sato's school of microlocal analysis (see for example \cite{sheavesonmanifolds}). This is not the point of view we discuss here but the interested reader is referred to \cite{konrad} where microlocal techniques combined with the theory of toric varieties are used to investigate Feynman integrals from the parametric viewpoint.\\
More commonly one takes the $\epsilon$ to be a small positive value moving the singularities off the integration domain. One possibility to proceed is to keep the $i\epsilon$-term until the calculation is finished and take the limit $\epsilon\to0^+$ at the end of the calculation. In our setup this is not possible since if we projectivize, the $i\epsilon$ acquires a homogenous factor of degree 2 in the auxiliary variable. \enquote{At infinity}, where this factor vanishes, the zero loci still intersect the integration domain. Therefore we need to apply a so called \textit{Wick rotation}\footnote{Unfortunately this name is not used consistently in the literature: While some authors call the deformation of the original integration cycle a Wick rotation \cite{weinzierl}, like we do here, others reserve this word for the analytic continuation from real external momenta to Minkowski space \cite{hyperlog}.} (which the authors of the unpublished paper in \cite{homfeynman} call \textit{Dyson change}):\\
Viewing $I_G$ as an iterated integral we can successively deform the integration contour $i\mathbb{R}$ of the loop momenta's 0th components to $\mathbb{R}$ without crossing any pole. The resulting integral, if it converges at all, is absolutely convergent for appropriate (Euclidean) external momenta and masses. We can thus view it as an integral over $\mathbb{R}^{Dh_1(G)}$ and take the limit $\epsilon\to0^+$ inside the integration. There is a little caveat: As the authors of the aforementioned unpublished paper remark, there is an ambiguity in equation (\ref{eq:feynman_integral}). We need to agree on the order in which we take the integrations and it is not clear if the results are equivalent. In any case the function defines a germ of an analytic function given by the Wick rotated integral but they might correspond to different branches.\\
We will not dwell on the details of this procedure here (the reader is referred to \cite{weinzierl}) and instead take the Wick rotated form to be the definition of a Feynman integral in this work (see Definition \ref{def:feynman_integral} below).\\
All of this does not mean, however, that we restrict ourselves to the case of Euclidean (which in our setup means real) external or internal momenta. The Wick rotated situation just serves as a starting point for an analytic continuation: Starting with Euclidean external momenta $p_0\in\mathbb{R}^{N_GD}$, real masses $m_0\in\mathbb{R}^{|E(G)|}$, and an integration domain consisting of all Euclidean internal momenta, we can analytically continue along appropriate paths in our parameter space (which we have not defined so far, but which generally consists of the (complex) external momenta and possibly the (complex) masses). The zero loci and the initial integration cycle $\Gamma_0=\mathbb{R}\mathbb{P}^{h_1(G)D}$ will eventually collide requiring us to deform $\Gamma_0$ to more general and in particular Minkowski momenta.

\subsection{Analyticity of Feynman Integrals}
After applying the Wick rotation a Feynman integral is in a form in which the results from Section \ref{sec:quadric_integrals} apply. The general Feynman integral, after projectivizing, is then written in the form of a projective quadratic integral. At this point we need to agree on which variables we want the Feynman integral, interpreted as a complex function, to depend on and which variables we view as fixed. The candidates for a functional dependence are the external momenta $p\in\mathbb{C}^{N_GD}$ and the masses $m\in\mathbb{C}^{|E(G)|}$. The experimental situation usually allows for a variation of the external kinematic (of course restricted by the machines at hand) but the masses are fixed by the particles occurring in a given theory. It can sometimes be useful to consider the masses as additional complex variables. In our approach we can freely fix an arbitrary number of external momenta (or even just individual components) and masses while defining the rest to form the parameter manifold. For concreteness we allow all external momenta and masses to vary in the following discussion. In other words the parameter manifold for the integral corresponding to a Feynman graph $G\in\mathcal{G}$ is
\begin{align*}
    T_G:=\mathbb{C}^{N_GD}\times\mathbb{C}^{|E(G)|}.
\end{align*}
We thus define a Feynman integral as follows:
\begin{defn}\label{def:feynman_integral}
    Let $G\in\mathcal{G}$ be a Feynman graph. We define the corresponding \textit{Feynman integral} $I_G$ to be
    \begin{align*}
        I_G(p,m):=\int_{\mathbb{R}^{h_1(G)D}}\frac{dy_1\wedge\cdots\wedge dy_{h_1(G)D}}{\prod_{e\in E}(Q_e(p,m)(y))^{\lambda_i}}.
    \end{align*}
    We call a Feynman integral $I_G$ \textit{convergent} if $I_G(p,m)$ converges as an integral for all $(p,m)$ in some open set $U\subset T_G$. If $2\text{Re}(\lambda)\geq h_1(G)D+1$, we say that $I_G$ is \textit{superficially convergent as a projective quadric integral}.
\end{defn}
\noindent Our definition of superficial convergence as a projective quadratic integral is a little different from the usual definition of superficial convergence in the physics literature, where a Feynman integral is called superficially convergent if $2\text{Re}(\lambda)>h_1(G)D$. We do make use of this second meaning of the term in this work, so there is no danger of confusion in simply saying $I_G$ is superficially convergent if $2\text{Re}(\lambda)\geq h_1(G)D+1$ throughout the rest of this text. Note that
\begin{align*}
    I_G(p,m):=\int_{\mathbb{R}\mathbb{P}^{h_1(G)D}}\frac{u^{2\lambda-h_1(G)D-1}\cdot\Omega_{h_1(G)D}}{\prod_{e\in E}(H(Q_e(p,m))(u,k))^{\lambda_e}}
\end{align*}
by Proposition \ref{thm:projectivize} and thus a superficially convergent Feynman integral as defined above constitutes a projective quadratic integral. Since it is customary to index the first component of each loop momentum by $0$, we denote the auxiliary variable introduced through compactification by $u$ instead of $k_0$. Furthermore we decompose $Q_e(p,m)(k)=(K_e(k)+P_e(p))^2+m_e^2$ where $K_e(k)\in\mathbb{C}^D$ contains the dependence on the internal and $P_e(p)\in\mathbb{C}^D$ the dependence on the external momenta associated to the edge $e$. We additionally write $K_e(k)=\sum_{i=1}^{h_1(G)}c^{(e)}_ik_i$.\\
Feynman integrals are sufficiently regular for our techniques to apply as the next proposition establishes.
\begin{prop}\label{thm:fully_regularizable}
    Let $G\in\mathcal{G}$ be a Feynman graph, $I_G$ the corresponding Feynman integral, and $(p,m)\in\mathbb{R}^{N_GD}\times\mathbb{R}^{|E(G)|}$. Then $I_G$ is quasi-regular at $(p,m)$.
\end{prop}
\begin{proof}
    Let $(p,m)\in\mathbb{R}^{N_GD}\times\mathbb{R}^{|E(G)|}$. Then from the general form
    \begin{align*}
        H(Q_e(p,m))(u,k)=(K_e(k)+uP_e(p))^2+u^2m_e^2
    \end{align*}
    of the quadratic forms $H(Q_e(p,m))$ we can see immediately that the corresponding matrix has real entries and that we have $H(Q_e(p,m))(u,k)\geq0$ for $(u,k)\in\mathbb{R}^{h_1(G)D+1}$ since $K_e(k)+uP_e(p)$ and $u^2m_e^2$ are real. Hence $H(Q_e(p,m))$ is positive semi-definite for all $e\in E(G)$.
\end{proof}
Some remarks on the $Q$-regularization of quadratic integrals in the case of Feynman integrals are in order: The method we employ is somewhat unorthodox and might seem strange to a physicist on first sight. The quadratic forms $R_\epsilon(H(Q_e(p,m)))$ contain a contribution from every loop momentum. The regularization \enquote{forgets} the structure of subdivergences which is now completely encoded in the matrices $A_e(p,m)$ from Lemma \ref{thm:key_lemma_2}. The $Q$-regularized integral $R_\epsilon(I_G)$ is already convergent if it is superficially convergent. While the subdivergencies play a crucial role in renormalization which we discuss in the Subsection \ref{sec:renormalization}, they are not relevant in our discussion of the Landau equations here. This is precisely the idea of this technique: We decompose the study of a Feynman integral into two steps, where in the first step we forget about the complicated structure of nested divergencies and in the second step take the limit $\epsilon\to0$. In Theorem \ref{thm:quadric_integrals_main_theorem} we required the unregularized integral to converge, but to achieve convergence as $\epsilon\to0$ we actually need to take care of the divergencies. This can be achieved by \textit{renormalization} which is discussed in Section \ref{sec:renormalization}.
\begin{example}
    Consider the following Feynman graph (often called the sunrise or sunset graph):
    
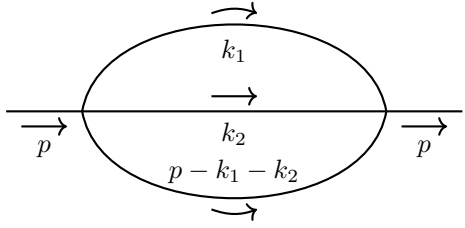
\begin{center}
    \begin{tikzpicture}[thick]
    \draw (-2,0) to [out=80, in=100] (2,0);
    \draw [->] (-2.8,-0.2) -- (-2.2,-0.2);
    \draw node at (-2.5,-0.5) {$p$};
    \draw (-2,0) to [out=-80, in=-100] (2,0);
    \draw [->] (-0.3,1.3) to [out=20, in=160] (0.3,1.3);
    \draw node at (0,0.8) {$k_1$};
    \draw [->] (-0.3,-1.3) to [out=-20, in=-160] (0.3,-1.3);
    \draw node at (0,-0.8) {$p-k_1-k_2$};
    \draw (-2,0) to (2,0);
    \draw [->] (-0.3,0.2) to (0.3,0.2);
    \draw node at (0,-0.3) {$k_2$};
    \draw (-3,0) -- (-2,0);
    \draw (2,0) -- (3,0);
    \draw [->] (2.2,-0.2) -- (2.8,-0.2);
    \draw node at (2.5,-0.5) {$p$};
    \end{tikzpicture}
    \end{center}
    Here $k_1$ and $k_2$ are the internal momenta, $p$ is the external momentum (note that we already applied the overall momentum conservation). In this picture each edge $e$ is labeled with $K_e(k)+P_e(p)$, the momentum assigned to it after enforcing momentum conservation at the two vertices. The orientation used is indicated by the arrows next to the edges. The three propagators (in the chart with $u\neq0$) read
    \begin{align*}
        Q_1(p,m)(k)=\frac{1}{k_1^2+m_1^2},\quad Q_2(p,m)(k)=\frac{1}{k_2^2+m_2^2},\quad Q_3(p,m)(k)=\frac{1}{(p-k_1-k_2)^2+m_3^2}.
    \end{align*}
    The matrices from Lemma \ref{thm:key_lemma_2} are not uniquely determined. One could for example add the positive definite form given by $(u,k)\mapsto\epsilon\cdot(u^2+k^2)$ to each denominator which is arguably the easiest possible choice. Assuming we exclude vanishing masses from our parameter space we could also choose the following $Q$-regularization:
    \begin{gather*}
        R_\epsilon(H(Q_1(p,m)(u,k)))=\frac{1}{k_1^2+\epsilon k_2^2+m_1^2u^2}, \quad R_\epsilon(H(Q_2(p,m)(u,k)))=\frac{1}{\epsilon k_1^2+k_2^2+m_2^2u^2},\\ R_\epsilon(H(Q_1(p,m)(u,k)))=\frac{1}{(up-k_1-k_2)^2+\epsilon(k_1^2+k_2^2)+m_3^2u^2}.
    \end{gather*}
    Note how each propagator now depends non-trivially on all loop-momenta occurring in the integral.
\end{example}
By the discussion in Section \ref{sec:quadric_integrals}, particularly Theorem \ref{thm:quadric_integrals_main_theorem}, we immediately obtain the following from Proposition \ref{thm:fully_regularizable}:
\begin{thm}
    Let $G\in\mathcal{G}$ and $I_G$ be the corresponding Feynman integral. If $I_G$ is convergent it defines a multi-valued function on $\mathbb{C}^m\backslash L(I_G)$.
\end{thm}
This means in particular that all points at which a Feynman integrals $I_G$ might not be analytic are determined by the Landau equations (\ref{eq:landau_equations}). We can separate them into equations for $L_\text{fin}(I_G)$ (by choosing non-homogeneous coordinates) and for $L_\infty(I_G)$ (by setting $u$ to 0). For $L_\text{fin}(I_G)$ the equations then read
\begin{gather*}
    \alpha_i=0\quad\text{or}\quad Q_i(p,m)(k)=0\quad\forall1\leq i\leq |E(G)|\\
    \sum_{i=1}^{|E(G)|}\alpha_i\frac{\partial}{\partial k_j}Q_i(p,m)(k)=0\quad\forall1\leq j\leq h_1(G)
\end{gather*}
which is a slightly more general form of the momentum-space Landau equations than they are written in a typical textbook on quantum field theory (due to $\alpha$ and $k$ being allowed to take arbitrary values in $\mathbb{C}^{|E(G)|}\backslash\{0\}$ and $\mathbb{C}^{h_1(G)D}$ respectively). We refer to them here as \textit{finite Landau equations}. Note that we have dropped the equation belonging to the component $u$ since it is superfluous. To see this note that if $Q_i(p,m)(k)=0$ then
\begin{align*}
    \sum_{j=1}^{h_1(G)}k_j\frac{\partial}{\partial k_j}Q_i(p,m)(k)=-u\frac{\partial}{\partial u}Q_i(p,m)(k)
\end{align*}
and hence if there are $\alpha_1,\ldots,\alpha_E\in\mathbb{C}$ such that $\sum_{i=1}^{|E(G)|}\alpha_i\frac{\partial}{\partial k_j}Q_i(p,m)(k)=0$ for all $j\in\{1,\ldots,h_1(G)\}$ then we automatically get
\begin{align*}
    0=\sum_{j=1}^{h_1(G)}k_j(\sum_{i=1}^{|E(G)|}\alpha_i\frac{\partial}{\partial k_j}Q_i(p,m)(k))=-u\sum_{i=1}^{|E(G)|}\alpha_i\frac{\partial}{\partial u}Q_i(p,m)(k)
\end{align*}
which implies the claim since $u\neq0$. The equations for $L_\infty(I_G)$ take the form
\begin{equation}\label{eq:landau_second_type}
    \begin{gathered}
        \alpha_i=0\quad\text{or}\quad (K_i(k))^2=0\quad\forall1\leq i\leq |E(G)|\\
        \sum_{i=1}^{|E(G)|}\alpha_ic^{(i)}_jK_i(k)=0\quad\forall1\leq j\leq h_1(G)\\
        \sum_{i=1}^{|E(G)|}\alpha_iP_i(p)K_i(k)=0.
    \end{gathered}
\end{equation}
As it stands this set of equations is not very useful for practical calculations in physics since there are too many solutions at non-physical points. Particularly the equations for the second type should be taken with a grain of salt. But of course being a solution to the Landau equations is merely a necessary criterion for the integral to be non-analytic, not a sufficient one. In case of second type singularities we consider the zero loci at infinity, where they do not depend on $p$ or $m$ at all.\footnote{This is not in contradiction with the dependence on $(p,m)$ occurring in the Landau equations of the second type: Which points in $\mathbb{C}\mathbb{P}^{h_1(G)D}$ with $u=0$ belong to the zero loci does not depend on $(p,m)$, but the gradients at these points do.} So if they are initially disjoint from the integration cycle, no pinch can occur until the integration cycle is somehow deformed to meet the zero loci at infinity.\\
In accordance with the special theory of relativity, physicists are typically interested in specific values for the momenta. These are the momenta in Minkowski space which in our setting are spanned by all vectors such that the 0th component of every vector $p_i$ or $k_i$ is purely imaginary while the remaining components are real. Thus we define the subspaces of Minkowski momenta for a Feynman graph $G\in\mathcal{G}$ as
\begin{align*}
    M_p(G):=i\mathbb{R}\times\mathbb{R}^{D-1}\times\cdots\times i\mathbb{R}\times\mathbb{R}^{D-1}\subset\mathbb{C}^{N_GD}
\end{align*}
and
\begin{align*}
    M_k(G):=i\mathbb{R}\times\mathbb{R}^{D-1}\times\cdots\times i\mathbb{R}\times\mathbb{R}^{D-1}\subset\mathbb{C}^{h_1(G)D}.
\end{align*}

\begin{defn}
    We say that a point $(p,m)\in T_G$ is \textit{physical} if $p\in M_p$ with $p_i^2\leq0$ for all $i\in\{1,\ldots,N_G\}$ and $m\in\mathbb{R}^{|E(G)|}$. For a Feynman graph $G\in\mathcal{G}$ we define the \textit{physical Landau variety} of the corresponding Feynman integral $I_G$ by
    \begin{align*}
        L_\text{phys}(I_G):=\pi(c_\text{fin}(I_G)\cap((\mathbb{R}^{|E(G)|}_{\geq0}-\{0\})\times M_k(G)\times M_p(G)\times\mathbb{R}^{|E(G)|}))\subset L_\text{fin}(I_G).
    \end{align*}
\end{defn}
As indicated in Section \ref{sec:quadric_integrals} and implied by the above definition, it is well-known to physicists that non-analytic physical points on the principal branch should only occur as points in $L_\text{phys}(I_G)$. This is also the form of the Landau equations as they are typically stated in the literature. That this indeed describes all non-analytic points at physical points can not be directly seen using the techniques from \cite{app-iso} and we do not give a proof in this paper. The problem is that these techniques only consider the local geometry of the intersection of the zero loci. But even if the zero loci are not in general position for some $(p,m)\in T_G$, there might still be an appropriate integration cycle constituting an analytic continuation. A very simple example of this is $\int_\sigma\frac{dz}{z^2-t}$ where $\sigma$ is a small circle around $1$ (as in Subsection \ref{sec:integraltreatment}). The point $t=0$ lies in the Landau variety of this integral. If we choose $t_0=1$ as the starting point for our analytic continuation, this is indeed a non-analytic point. If we were to start at $t_0=\frac{1}{2}$ however we simply obtain $\int_\sigma\frac{dz}{z^2-t}=0$ in a neighborhood of $t_0$ and can analytically continue to $t=0$ without problems.\\
Although no proof of the analyticity of $I_G$ at physical points on $T_G\backslash L_\text{phys}(I_G)$ on the principal branch appears here, we nevertheless remark on a few points: If we already know that it suffices to consider solutions with $\alpha\in\mathbb{R}_{\geq0}^{|E(G)|}-\{0\}$, we can immediately restrict the space of relevant values of $k$ to the values appearing in $c_\text{phys}(I_G)$. Suppose $(p,m)\in T_G$ is a physical point and we have a solution $(\alpha,p,m,1,k)\in L(I_G)$. Then if $\pi:\mathbb{C}^{h_1(G)D}\twoheadrightarrow(M_k(G))^T$ denotes the canonical projection to the orthogonal complement of $M_k(G)$ (with respect to the standard scalar product) we have $\sum_{i=1}^{|E(G)|}\alpha_ic_j^{(i)}K_i(\pi(k))=0$ (since $P_i(p)\in i\mathbb{R}\times\mathbb{R}^{D-1}$ for $p\in M_p(G)$). The only solution to this system of linear equations is $\pi(k)=0$ hence we may always assume $k\in M_k(G)$ in this case. There are various ways to see this, all involving some level of graph theory. One easy argument is to note that the matrix corresponding to this linear system of equations satisfies
\begin{align*}
    \det\begin{pmatrix}\sum_{i=1}^{|E(G)|}\alpha_ic_1^{(i)}c_1^{(i)} & \cdots & \sum_{i=1}^{|E(G)|}\alpha_ic_1^{(i)}c_{h_1(G)}^{(i)}\\ \vdots & \ddots & \vdots \\ \sum_{i=1}^{|E(G)|}\alpha_ic_{h_1(G)}^{(i)}c_1^{(i)} & \cdots & \sum_{i=1}^{|E(G)|}\alpha_ic_{h_1(G)}^{(i)}c_{h_1(G)}^{(i)}\end{pmatrix}=\mathcal{U}_G(\alpha),
\end{align*}
which is the content of Theorem 3-10 in \cite{nakanishi}. Here $\mathcal{U}_G$ is the first Symanzik polynomial associated to $G$. It is well-known that $\mathcal{U}_G(\alpha)>0$ for $\alpha\in\mathbb{R}^{|E(G)|}_{\geq0}-\{0\}$ so that $\pi(k)=0$ is indeed the only solution to $\sum_{i=1}^{|E(G)|}\alpha_ic_j^{(i)}K_i(\pi(k))$. In a similar manner we see that there are no solutions to the equations of the second type for $\alpha\in\mathbb{R}_{\geq0}^{|E(G)|}-\{0\}$.

\section{Renormalization of Divergent Integrals}\label{sec:renormalization}
Until now we restricted the discussion to convergent integrals. To physicists it is a well-known phenomenon that the integrals of interest are often divergent and a procedure called \textit{renormalization} is employed to extract finite values. There exists extensive literature on the subject and we keep the discussion here short. The interested reader is referred to \cite{renormalization}.\\
Conceptually the need for renormalization in physics is not unique to quantum field theory and has been known for a long time. Trying to compute the self-interaction of a point charge in classical electrodynamics for example famously leads to a divergent result. The problem is that we compute such physical quantities by using parameters like the mass and charge of the electron in the example that we measure in experiments. The \textit{bare} parameters however might not be experimentally accessible. To account for this one needs to introduce renormalized parameters to salvage the situation. We know describe how renormalization manifests in perturbative quantum field theory.\\
To each Feynman graph $G\in\mathcal{G}$ we associate a number $\omega(G)=\sum_{i=1}^{|E(G)|}\text{Re}\,(\lambda_i)-\frac{D}{2}h_1(G)$\footnote{More generally one can also assign weights to the vertices of $G$: To each vertex $v\in V(G)$ we can associate a number $\omega(v)\in\mathbb{R}$. The generalized formula then reads $\omega(G)=\sum_{v\in V(G)}\omega(v)+\sum_{e\in E(G)}\lambda(e)-\frac{D}{2}h_1(G)$. It should also be noted that a lot of authors define the superficial degree of divergence with a relative sign with respect to our definition and/or include a factor of $2$ in their definition. Our convention is adopted from \cite{hyperlog}.}, the \textit{superficial degree of divergence} (compare to Definition \ref{def:quadric_integral} where we defined superficial convergence for quadratic integrals). Weinberg \cite{power-counting} famously proved that the Feynman integral corresponding to $G$ with all masses non-vanishing is convergent if and only if for each subgraph $\gamma\subset G$ (including $G$ itself) we have $\omega(\gamma)>0$. If $\omega(\gamma)\leq0$ it is said that $\gamma$ constitutes a \textit{subdivergency}. This consideration is known as \textit{power counting} and can informally be explained as follows: Each propagator $\frac{1}{Q_i^{\lambda_i}}$ contributes a term to the denominator which comes with a power of $2\lambda_i$ in the occurring loop-momenta, improving the convergence behavior (thus we count each edge with a weight of $2\text{Re}(\lambda_i)$) and each loop contributes $D$ integrations worsening the convergence. As long as we restrict ourselves to Euclidean external momenta, the Feynman integral converges absolutely if it converges at all and we may consider it as an iterated integral by Fubini's Theorem. Hence if the integral corresponding to a subgraph diverges so does the entire integral. The aim of renormalization is to systematically manipulate the integral to remove divergences in accordance with the structure of subdivergencies. A divergency arising from a failure of the power counting condition is called an \textit{ultraviolet divergency} (as they arise from large $k$) while the divergences arising from vanishing masses are called \textit{infrared divergences} (since they stem from small $k$). We consider here the renormalization of ultraviolet divergencies only since we consider the masses to be parameters. In this setup the infrared divergences arise as part of the Landau variety. However the infrared divergences can be treated by appropriate methods, see for example \cite{infrared}, although the issue is more subtle. The procedure of renormalization relies heavily on the graph structure and its underlying combinatorics. The set of Feynman graphs can be endowed with the structure of a Hopf algebra as discovered by Kreimer \cite{kreimer} which provides the necessary insight into the combinatorial structure of subdivergencies arising from any given graph.

\subsection{Regularization}\label{sec:feynman_regularization}
The typical renormalization procedure goes as follows: Introduce a \textit{regulator}, an auxiliary variable which recovers the original integral in some limit, for which a value exists such that the integral converges (outside of the Landau variety). Then systematically remove the nested divergences by manipulating the regulator-dependent result. Finally take the appropriate limit in the regulator to obtain a finite result.\\
There are many ways to implement regularization. Perhaps the most commonly known are the following three:
\begin{enumerate}
    \item The \textit{cut-off} method replaces the integration domain of the non-projectivized  Feynman integral by a bounded domain (say for example a sphere with radius $R$). The new integral does not suffer from ultraviolet divergences. After renormalization one considers the limit $R\to\infty$. Similarly one can introduce a non-zero lower bound on the absolute value of the internal momenta's components to remove infrared divergences. One of the major problems with this regularization is that it destroys the translation invariance of the integral.
    \item \textit{Dimensional regularization}\cite{dimreg2},\cite{dimreg} replaces the space-time dimension $D$ by a complex variable. Typically one fixes the dimension of interest $D$ and then replaced it by $D-2\epsilon$ where $\epsilon\in\mathbb{C}$. After renormalization one considers the limit $\epsilon\to0$. In the parametric representation $D$ simply appears as a variable and it is straightforward to perform this replacement. In the momentum space setting, defining the meaning of this idea is non-trivial. Since we do not use this approach in this text we do not explain the construction here. The curious reader is referred to one of the two papers above.
    \item \textit{Analytic regularization} \cite{analyticregularization}, which we already prepared in our setup, is the regularization by use of the $\lambda_i$: One is usually interested in the value of a Feynman integral for a specific choice of $\lambda_i\in\mathbb{N}^+$, most commonly $\lambda_1=\cdots=\lambda_N=1$. So we already implicitly analytically regularized our integrals by defining the Feynman integral with generic $\lambda_i$. It is convenient to fix the desired values $\lambda_1,\ldots,\lambda_N$ at which we wish to evaluate the integral and replace them by $\lambda_i+\epsilon$ with $\epsilon\in\mathbb{C}$ instead. Then we only need to take one limit $\epsilon\to0$ instead of $N$ limits in the $\lambda_i$.
\end{enumerate}
In this paper we exclusively employ analytic regularization. Some renormalization schemes like kinematic renormalization, which is explicitly discussed in Subsubsection \ref{sec:schemes}, can be performed on the level of the integrand and thus do not require regularization. To remain general and for a simpler treatment we still regularize even in this case.\\
For a general regularization in a theory which allows vanishing masses it is usually a highly non-trivial question if, for a given integral, the regulator can be chosen in such a way that the integral converges. In analytic regularization for example large real parts of the $\lambda_i$ certainly improve the convergence of the integral as the internal momenta become large. On the other hand, if we are dealing with massless theories, this might simultaneously worsen the infrared behavior, i.e. the behavior of the integrand as the internal momenta become small. We do not discuss this problem here. In the case of analytic regularization an exhaustive discussion from the parametric viewpoint can be found in \cite{konrad}.

\subsection{Hopf Algebraic Renormalization}
We review the necessary theory of Hopf algebras needed to formulate renormalization. All results introduced here are standard. For more details on the subject we recommend the excellent book \cite{hopfalgebras} by Sweedler or the more modern exposition \cite{hopfalgebras2}. All algebras (resp. co-algebras) in this section will be associative and unital (resp. co-associative and co-unital). First we recall the definition of a Hopf algebra:
\begin{defn}[\cite{hopfalgebras2}]
    A \textit{Hopf algebra} $(H,m,\Delta,\mathbbm{1},\hat{\mathbbm{1}},S)$\footnote{Typically the unit and counit are denoted by $\eta$ (or sometimes $u$) and $\epsilon$ respectively. In this section we want to reserve $\epsilon$ for the parameter introduced through analytic regularization.} is a bialgebra $(H,m,\Delta,\mathbbm{1},\hat{\mathbbm{1}})$ over a field $K$, where $m$ denotes multiplication, $\Delta$ co-multiplication, $\mathbbm{1}$ the unit, and $\hat{\mathbbm{1}}$ the co-unit, together with a map $S:H\to H$ called the \textit{antipode} such that the diagram
    \begin{center}
		\begin{tikzcd}[row sep=3.6em,column sep=1em]
			& H\otimes H \arrow[rr,"S\otimes\text{id}"] && H\otimes H \arrow[dr,"m"] \\
			H \arrow[ur,"\Delta"] \arrow[rr,"\hat{\mathbbm{1}}"] \arrow[dr,"\Delta"'] && K \arrow[rr,"\mathbbm{1}"] && H \\
			& H\otimes H \arrow[rr,"\text{id}\otimes S"'] && H\otimes H \arrow[ur,"m"']
		\end{tikzcd}
	\end{center}
    commutes.
\end{defn}
Let $(H,m,\Delta,\mathbbm{1},\hat{\mathbbm{1}},S)$ be a Hopf algebra. We denote $1_H:=\mathbbm{1}(1)$. The kernel $\ker\hat{\mathbbm{1}}$ is called the \textit{augmentation ideal}. We say $H$ is \textit{connected} if it admits a filtration $H=\sum_{n\in\mathbb{N}}H_n$ with $H_0=K\cdot 1_H$.\\
Let $A$ be another algebra with multiplication $m_A$ and unit $1_A\in A$. For two homomorphisms $\phi_1,\phi_2\in\text{Hom}(H,A)$ we denote their \textit{convolution product} by
\begin{align*}
    \phi_1\star\phi_2:=m_A\circ(\phi_1\otimes\phi_2)\circ\Delta\in\text{Hom}(H,A).
\end{align*}
We denote the inverse of a linear map $\phi:H\to A$ with respect to the convolution product, if it exists, by $\phi^{\star-1}$. The elements invertible with respect to convolution form a group under this operation. In particular recall that if $H$ is connected, the linear maps sending $1_H$ to $1_A$ form a subgroup $G^H_A$ of this group. Moreover if $A$ is commutative the \textit{characters}, i.e. morphisms $\phi:H\to A$ of unital algebras, belong to this group and indeed form a subgroup of $G^H_A$. The inverse of a character $\phi$ is then given by the antipode as $\phi^{\star-1}=\phi\circ S$. A bialgebra $H$ is Hopf if and only if the identity $\text{id}_H:H\to H$ admits an inverse with respect to the convolution product. This inverse must then be the anitpode, i.e. $S=\text{id}_H^{\star-1}$.\\
There are several ways to endow the set of graphs under consideration with the structure of a Hopf algebra. Of course different theories have different underlying sets of admissible graphs. Even if we agree on an appropriate set of graphs, there are still various coproducts to choose from which determine which graphs we consider to be in need of renormalization. These choices do not affect our argument below. For a detailed discussion on this setup in more specific settings see for example \cite{david}. Since we do not consider a specific theory but rather allow arbitrary graphs we stick to the following Hopf algebra for concreteness:\\
We set
\begin{align*}
    \mathcal{G}_\text{1PI}:=\{G\in\mathcal{G} \;|\; G\text{ is 1PI}\}
\end{align*}
and let $H_G:=\text{span}_\mathbb{Q}\;\mathcal{G}_\text{1PI}$ as a $\mathbb{Q}$-vector space.\footnote{We could choose a different ground field of characteristic 0.} Then we define the multiplication $m:H_G\otimes H_G\to H_G$ by disjoint union and the coproduct $\Delta:H_G\to H_G\otimes H_G$ by
\begin{align*}
    \Delta(G):=G\otimes1_{H_G}+1_{H_G}\otimes G+\sum_{\substack{\gamma\subsetneq G \\ \omega(\gamma)\leq 0}}\gamma\otimes G/\gamma.
\end{align*}
It is not difficult to verify that $\Delta$ is coassociative turning $H_G$ into a bialgebra. We equip $H_G$ with a grading given by the first Betti number, i.e. $H_G=\bigoplus_{n\in\mathbb{N}}H_G^{(n)}$ where
\begin{align*}
    H_G^{(n)}:=\{G\in H_G \;|\; h_1(G)=n\}.
\end{align*}
The multiplication $m$ and co-multiplication $\Delta$ are compatible with and $H_G$ is connected with respect to this grading (recall that a grading of a Hopf algebra induces a filtration). Any connected bialgebra is Hopf and this holds in particular for $H_G$.\\
We need to choose an appropriate target algebra $A_\text{target}$ for our Feynman rules which we want to interpret as a character $H_G\to A_\text{target}$. We want to consider the result of applying the Feynman rules to a Feynman graph $G\in\mathcal{G}$ as a meromorphic function, i.e. a function that can locally be written as the quotient of two holomorphic functions, in the regulator $\epsilon\in\mathbb{C}$ with coefficients functions in the external momenta and masses in an appropriate ring of functions. How we choose this ring exactly is of little relevance in our context as long as it is big enough to contain all functions arising from Feynman integrals in this way. For concreteness we denote for a finite-dimensional complex analytic manifold $M$ the ring of all holomorphic functions defined on an open and dense subset of $M$ by $\mathcal{F}(M)$ with the obvious addition and multiplication. Then for a given space-time dimension $D$ we define the target algebra to be
\begin{align*}
    A_\text{target}:=(\bigoplus_{n,m\in\mathbb{N}} \mathcal{F}(\mathbb{C}^{nD}\times\mathbb{C}^m))[\epsilon^{-1},\epsilon]],
\end{align*}
where $R[x^{-1},x]]$ denotes the ring of Laurent-series in $x$ with coefficients in a ring $R$.
Feynman rules can then be interpreted as a character $H_G\to A_\text{target}$:
\begin{defn}
    We define the \textit{Feynman rules} $\phi_\text{Feyn}:H_G\to A_\text{target}$ by
    \begin{align*}
        \phi_\text{Feyn}(G):=\mathcal{L}(I_G)\in\mathcal{F}(\mathbb{C}^{N_GD}\times\mathbb{C}^{|E(G)|})[\epsilon^{-1},\epsilon]]\subset A_\text{target}
    \end{align*}
    for all connected $G\in\mathcal{G}$ where $\mathcal{L}(I_G)$ is the Laurent expansion of $I_G$ with respect to the regulator $\epsilon$.
\end{defn}
In this setting renormalization can be understood as a form of algebraic \textit{Birkhoff decomposition} \cite{conneskreimer}: Let $A_\text{target}=A_+\oplus A_-$ be a decomposition of $A_\text{target}$ into two algebras $A_\pm$ (not necessarily unital). Then, since $H_G$ is connected, any linear map preserving the unit admits a unique Birkhoff decomposition, i.e. there exist uniquely determined linear and unit preserving maps $\phi_+,\phi_-:H_G\to A$ such that $\phi=\phi_-^{\star-1}\star\phi_+$ and $\phi_\pm(\ker\hat{\mathbbm{1}})\subset A_\pm$. The decomposition then defines a \textit{renormalization scheme} $R:A\twoheadrightarrow A_-$, where $R$ is the canonical projection. If $\phi_\text{Feyn}$ are the Feynman rules then $(\phi_\text{Feyn})_+$ is called the \textit{renormalized Feynman rules} and $(\phi_\text{Feyn})_-$ the \textit{counter term}. If $x\in\ker\hat{\mathbbm{1}}$ we can obtain $\phi_\pm$ recursively from $\phi$ as
\begin{align}\label{eq:renormalization_recursion}
    \phi_-(x)=-R(\bar{\phi}(x)),\qquad\phi_+(x)=(\text{id}-R)(\bar{\phi}(x))
\end{align}
via the \textit{Bogoliubov map}
\begin{align}\label{eq:bogoliubov}
    \bar{\phi}(x):=\phi(x)+m(\phi_-\otimes\phi)\tilde{\Delta}(x)=\phi_+(x)-\phi_-(x),
\end{align}
where $\tilde{\Delta}:=\Delta-1_{H_G}\otimes\text{id}-\text{id}\otimes1_{H_G}$ is the reduced coproduct. Finally the \textit{physical limit} $\phi_R$ is defined by $\phi_R(G):=\lim_{\epsilon\to 0^+}\phi_+(G)(\epsilon)$ for all $G\in H_G$. This prescription does certainly not make sense for every renormalization scheme $R$. One needs to prove the existence of the physical limit for a given scheme. Such a proof is beyond the scope of this paper. In the following we introduce two specific renormalization schemes and give references to the appropriate literature.

\subsubsection{Two Renormalization Schemes}\label{sec:schemes}
Various renormalization schemes are used in practice. For a more concrete result we explicitly consider two schemes regularly used in practice: \textit{kinematic renormalization}, also known as the MOM-scheme, and \textit{minimal subtraction}.\\
\paragraph{\textbf{Kinematic renormalization}}
The idea of kinematic renormalization is to subtract part of the Taylor series at a specific kinematic reference point. For simplicity we restrict to the case of \textit{logarithmically divergent} graphs, i.e. those with $\omega(G)=0$. We fix a family $\{\mu_k\}_{k=0}^\infty$ of reference momenta $\mu_k\in\mathbb{C}^{kD}$. Then we define kinematic renormalization by the renormalization scheme
\begin{align*}
    R_\text{MOM}:A\to A,\quad f\mapsto f|_{p=\mu_{N_G}}.
\end{align*}
This looks slightly different than the usual definition of the kinematic renormalization scheme, since one typically makes use of the Lorentz-invariance of Feynman integrals, reducing the parameter space of a Feynman graph $G\in\mathcal{G}$ from $T$ to $(\mathbb{C}^{N_GD}/\mathcal{O}(N_GD,\mathbb{C}))\times\mathbb{C}^{|E(G)|}$, where $\mathcal{O}(N_GD,\mathbb{C})$ is the Lie group of complex orthogonal matrices of size $N_GD\times N_GD$ (or even further if one restricts to the hyperplane defined by overall momentum conservation). One useful way to encode this parameter space is to fix a \textit{scale} $s$ and express the parameters as multiples of this scale via \textit{angles} $\Theta$. Then kinematic renormalization becomes the evaluation at a reference scale $\tilde{s}$ and reference angles $\tilde{\Theta}$. We do not discuss this reduced parameter space and its possible implications here. An extensive discussion of kinematic renormalization in this setting can be found in \cite{brownkreimer}.

\paragraph{\textbf{Minimal subtraction}}
As the name implies, minimal subtraction removes the bare minimum to obtain a convergent integral.  It removes precisely the principal part of the Laurent series in $\epsilon$, i.e. we define the renormalization scheme by
\begin{align*}
    R_\text{min}:A\to A,\quad \sum_{k=-\infty}^\infty f_k(z_1,\ldots,z_n)\epsilon^k\mapsto \sum_{k=-\infty}^{-1}f_k(z_1,\ldots,z_n)\epsilon^k.
\end{align*}
This means the decomposition $A_\text{target}=A_+\oplus A_-$ in this scheme reads
\begin{align*}
    A_+=(\bigoplus_{n,m\in\mathbb{N}} \mathcal{F}(\mathbb{C}^{nD}\times\mathbb{C}^m))[[\epsilon]]\qquad\text{and}\qquad A_-=\epsilon^{-1}(\bigoplus_{n,m\in\mathbb{N}} \mathcal{F}(\mathbb{C}^{nD}\times\mathbb{C}^m))[[\epsilon^{-1}]].
\end{align*}
A more detailed discussion can be found in \cite{panzerrenorm}.

\subsection{Compatibility with the Landau Equations}
The main point of this subsection is to show that the previous discussion regarding the Landau variety is compatible with common forms of renormalization. We assume throughout this part that the physical limit of the renormalization scheme at hand is known to exist. We generally want the renormalized Feynman integral to define a function which is analytic outside of its Landau variety. Hence the following definition:
\begin{defn}
    Let $A_\text{target}=A_+\oplus A_-$ be a decomposition of the target algebra and $\phi_\pm$ the corresponding Birkhoff decomposition of the Feynman rules $\phi_\text{Feyn}$. The decomposition is called \textit{Landau compatible} if for every $G\in H_G$ the physical limit $\phi_R(G)$ defines an analytic function on $T_G\backslash L(I_G)$.
\end{defn}
Applying $\phi_+$ to a graph $G\in\mathcal{G}$ yields a Laurent series with coefficients functions in the external momenta and masses. If the physical limit exists, $\phi_+(G)$ has trivial principal part and $\phi_R(G)$ is just the coefficient of order 0 in $\phi_+(G)$. So the question of Landau compatibility is the question whether this coefficient is an analytic function on $T_G\backslash L(I_G)$ or not. From the previous section we know that the Laurent series $\phi_\text{Feyn}(G)(\epsilon)$ converges if $\epsilon$ has a sufficiently large real part and the coefficients of $\phi_\text{Feyn}(G)$ are analytic functions $T_G\backslash L(I_G)\to\mathbb{C}$. We do not go into the details of investigating when an arbitrary renormalization scheme is Landau compatible and postpone a more elaborate study to the future. We conclude this section however by remarking that the two renormalization schemes introduced above, where this can be seen rather easily, do in fact satisfy this condition:
\begin{prop}\label{thm:renormalization}
    The renormalization schemes $R_\text{MOM}$ and $R_\text{min}$ are Landau compatible.
\end{prop}
\begin{proof}
    It suffices to prove the result for $G\in\mathcal{G}$. We proceed by induction on the first Betti number $h_1(G)$. For $G$ a tree the result is obvious. Otherwise we see from equations (\ref{eq:renormalization_recursion}) and (\ref{eq:bogoliubov}) that $\phi_+(G)$ is a sum of $\phi_\text{Feyn}(G)$ itself as well as proper subgraphs and contracted graphs with $(\text{id}-R)$ applied to them (with $R=R_\text{MOM}$ or $R=R_\text{min}$). Applying $(\text{id}-R_\text{min})$ to $\phi_\text{Feyn}(G)$ and then taking the limit $\epsilon\to0$ simply extracts the 0th order coefficient from the series which is a function analytic on $T_G\backslash L(I_G)$ by the previous section. Applying $(\text{id}-R_\text{MOM})$ modifies the coefficients by subtracting a term constant in the kinematic variables, so that the 0th coefficient remains analytic outside $L(I_G)$. The remaining terms consist of graphs with stricly smaller Betti number and hence the induction hypothesis applies to them. We conclude that $\phi_{R_\text{MOM}}(G)$ and $\phi_{R_\text{min}}(G)$ indeed define analytic functions on $T_G\backslash L(I_G)$.
\end{proof}

\section{Examples}\label{sec:examples}
This section displays three simple examples of Landau surfaces including a simple Feynman integral. We remind the reader of our convention to drop the auxiliary integral $\int_\mathbb{R}\frac{dx}{x^2+\pi^2}$ from our notation, which is implicitly included in all integrals which would otherwise have non-orientable integration domain.
\subsection{A Very Simple Example}
Consider the quadratic integral
\begin{align*}
    I(t)=\int_\mathbb{R}\frac{dy}{y^2+t^2}=\int_{\mathbb{R}\mathbb{P}}\frac{\Omega_1}{t^2z_0^2+z_1^2}.
\end{align*}
The matrix representing the quadratic form in the denominator is
\begin{align*}
    \begin{pmatrix}
        t^2 & 0 \\ 0 & 1
    \end{pmatrix}.
\end{align*}
It has full rank for all $t\neq0$ and is positive definite if $t^2\in\mathbb{R}^+$. Hence it is regular on $\mathbb{R}^\times$ and in particular of standard form. The Landau equations (\ref{eq:landau_equations}) in this simple case read
\begin{align*}
    t^2z_0^2+z_1^2=0\quad\text{and}\quad \begin{pmatrix} t^2z_0 \\ z_1\end{pmatrix}=0
\end{align*}
and we immediately verify that the only solutions $(z_0,z_1,t)$ are of the form $(z_0,0,0)$ with $z_0\in\mathbb{C}^\times$ (which are of course all the same point in $\mathbb{C}\mathbb{P}$) and hence the Landau variety is simply $L(I)=\{0\}$. It decomposes as $L(I)=L_\text{fin}(I)\cup L_\infty(I)$ with $L_\text{fin}(I)=\{0\}$ and $L_\infty(I)=\emptyset$ (there are no singularities of the second type). This is in accordance with the fact that the given integral can, by explicit integration, easily be verified to define a germ of the function $\pi\sqrt{\frac{1}{t^2}}$ for all $t$ with $\text{Re }t>0$ which can be analytically continued along any path in $\mathbb{C}^\times$.
\subsection{A Slightly More Complicated Example}
We consider the projective quadratic integral
\begin{align*}
    I(t)=\int_{\mathbb{R}^2}\frac{dz_1\wedge dz_2}{(z_1^2+z_2^2+z_1t+1)(z_1^2+z_2^2+t^2)}=\int_{\mathbb{R}\mathbb{P}^2}\frac{z_0\cdot\Omega_2}{(z_1^2+z_2^2+z_0z_1t+z_0^2)(z_1^2+z_2^2+z_0^2t^2)}.
\end{align*}
Note that this expression does not correspond to any Feynman integral. The matrices representing the quadratic forms are
\begin{align*}
    M_1(t)=\begin{pmatrix}1 & \frac{t}{2} & 0 \\ \frac{t}{2} & 1 & 0 \\ 0 & 0 & 1\end{pmatrix},\qquad
    M_2(t)=\begin{pmatrix}t^2 & 0 & 0 \\ 0 & 1 & 0 \\ 0 & 0 & 1\end{pmatrix}.
\end{align*}
The finite Landau equations read
\begin{gather*}
    \alpha_1=0\text{ or }z_1^2+z_2^2+z_1t+1=0, \quad \alpha_2=0\text{ or }z_1^2+z_2^2+t^2=0 \\
    \alpha_1\begin{pmatrix} 2z_1+t \\ 2z_2 \end{pmatrix}+\alpha_2\begin{pmatrix} 2z_1 \\ 2z_2 \end{pmatrix}=0.
\end{gather*}
So $\alpha_1=-\alpha_2$ or $z_2=0$ from the second component of the last line. In the first case we get $t=0$ from the first component. In the second case either $\alpha_1=0$ (then $z_1=0$ and we get $t=0$ again), $\alpha_2=0$ (then $z_1=-\frac{t}{2}$ which implies $t=\pm2$), or $\alpha_1,\alpha_2\neq0$ in which case $z_1=\pm it$ and
\begin{align*}
    -z_1^2-z_1t-1=0\qquad\Rightarrow\qquad t^2=\frac{1}{2}(1\pm i).
\end{align*}
The Landau equations of the second type read
\begin{align*}
    z_1^2+z_2^2=0,\qquad\alpha_1\begin{pmatrix} z_1t \\ 2z_1 \\ 2z_2\end{pmatrix}+\alpha_2\begin{pmatrix} 0 \\ 2z_1 \\ 2z_2\end{pmatrix}=0.
\end{align*}
Since we can not have $z_1=z_2=0$ we obtain $\alpha_1=-\alpha_2$ from the last two components of the second equation. This means $\alpha_1,\alpha_2\neq0$ so the first component gives $z_1t=0$. Since $z_1=0$ implies $z_2=0$ by the first equation we get yet again $t=0$. We conclude that
\begin{align*}
    L(I)=\{t\in\mathbb{C} \;|\; t=0\text{ or }t^2=\frac{1}{2}(1\pm i)\text{ or }t=\pm2\},
\end{align*}
which decomposes as $L(I)=L_\text{fin}(I)\cup L_\infty(I)$ with
\begin{align*}
    L_\text{fin}=L(I)\quad\text{and}\quad L_\infty=\{0\}.
\end{align*}
Note that we have a singularity of mixed type at $t=0$. The result is in accordance with the integral representation
\begin{align*}
    I(t)=\frac{\pi i}{\sqrt{-2t^4+2t^2-1}}(&\ln(-t^2-i\sqrt{-2t^4+2t^2-1}+1)-\ln(t^2-i\sqrt{-2t^4+2t^2-1}-1)\\
    &-\ln(-3t^2-2i\sqrt{-2t^4+2t^2-1}+2)+\ln(3t^2-2i\sqrt{-2t^4+2t^2-1}-2))
\end{align*}
for $I$ which can be obtained by (somewhat tedious) direct integration and holds for all $t$ such that $\text{Re}(t)\neq0$ and $\text{Im}(t)\neq\pm\text{Im}(\sqrt{t^2-4})$.

\subsection{The Bubble-Graph}
The easiest while still interesting example of a Feynman integral in a scalar field theory is arguably the one associated to the following Feynman graph:
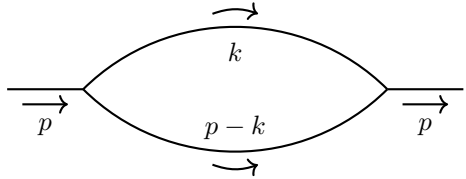
\begin{center}
\begin{tikzpicture}[thick]
\draw (-2,0) to [out=45, in=135] (2,0);
\draw [->] (-2.8,-0.2) -- (-2.2,-0.2);
\draw node at (-2.5,-0.5) {$p$};
\draw (-2,0) to [out=-45, in=-135] (2,0);
\draw [->] (-0.3,1.0) to [out=20, in=160] (0.3,1.0);
\draw node at (0,0.5) {$k$};
\draw [->] (-0.3,-1.0) to [out=-20, in=-160] (0.3,-1.0);
\draw node at (0,-0.5) {$p - k$};
\draw (-3,0) -- (-2,0);
\draw (2,0) -- (3,0);
\draw [->] (2.2,-0.2) -- (2.8,-0.2);
\draw node at (2.5,-0.5) {$p$};
\end{tikzpicture}
\end{center}
We denote this graph by $B_2$. The labeling of the edges indicates the momenta associated to each edge and vertex after enforcing momentum conservation in the obvious manner. In a typical physics textbook the associated Feynman integral in $D=4$ space-time dimensions is written as
\begin{align}\label{eq:onelooptwopoint}
    I(p,m_1,m_2)=\int_{\mathbb{R}^3}d^3\vec{k}\int_{-i\infty}^{i\infty}dk_0\frac{1}{(k^2+m_1^2-i\epsilon)^{\lambda_1}((p-k)^2+m_2^2-i\epsilon)^{\lambda_2}}.
\end{align}
This integral diverges for $\lambda_1=\lambda_2=1$ and requires renormalization (it is logarithmically divergent: $\omega(B_2)=2\cdot2-4\cdot1=0$). We will not explicitly go through the steps of renormalizing this example but we have seen in Section \ref{sec:renormalization} that, using an appropriate renormalization scheme, no additional problematic points occur in the renormalized function. In our setup (after performing a Wick rotation and projectivizing) the corresponding Feynman integral reads
\begin{align*}
    I(p,m_1,m_2)=\int_{\mathbb{R}\mathbb{P}^4}\frac{u^{2(\lambda_1+\lambda_2)-5}\cdot\Omega_4}{(k^2+m_1^2u^2)^{\lambda_1}((up-k)^2+m_2^2u^2)^{\lambda_2}}.
\end{align*}
The two quadratic forms corresponding to the two propagators are represented by two $5\times5$ matrices
\begin{align*}
    M_1(p,m)=\begin{pmatrix} m_1^2 & 0 \\ 0 & 1_4\end{pmatrix},\qquad M_2(p,m)=\begin{pmatrix} p^2+m_2^2 & -p^T \\ -p & 1_4 \end{pmatrix}
\end{align*}
where again $1_n$ denotes the identity matrix of size $n\times n$. The condition on the linear dependence of the gradients reads
\begin{align*}
    \alpha_1\begin{pmatrix} m_1^2u \\ k \end{pmatrix}+\alpha_2\begin{pmatrix} -pk+(p^2+m_2^2)u \\ k-pu \end{pmatrix}=0.
\end{align*}
The only solutions to the Landau equations (\ref{eq:landau_equations}) with one of the $\alpha_i$ vanishing must have $m_1=0$ (for $\alpha_2=0$) or $m_2=0$ (for $\alpha_1=0$). If $(\alpha_1,\alpha_2)\neq(0,0)$ we have either $u=k^2=0$ or $u\neq0$. For the first case, the case of second type singularities, we get $\alpha_1=-\alpha_2$ from the last 4 components (since in this case we have $k\neq0$). So we are looking for solutions to $pk=k^2=0$ where $k\neq0$. Denote $k=(k_0,\vec{k})$ and $p=(p_0,\vec{p})$. Either $\vec{k}^2=0$ which yields $k_0=p_0=0$ and $\vec{p}\cdot\vec{k}=0$, or $\vec{k}^2\neq0$ in which case the general solution can be expressed as $p_0=\pm i\frac{\vec{k}\cdot\vec{p}}{\sqrt{\vec{k}^2}}$. Note that if we assume $k\in M_k(B_2)$ and $p\in M_p(B_2)$, this solution satisfies
\begin{align*}
    p^2=-\frac{(\vec{k}\cdot\vec{p})^2}{\vec{k}^2}+\vec{p}^2\geq -\frac{\vec{k}^2\vec{p}^2}{\vec{k}^2}+\vec{p}^2=0.
\end{align*}
Hence the only solution with $(p,m)$ a physical point is $p^2=0$. In the second case $u\neq0$ we get
\begin{align*}
    k'^2=-m_1^2\quad\text{and}\quad k'p=\frac{1}{2}(p^2-m_1^2+m_2^2),
\end{align*}
where we denoted $k':=\frac{k}{u}$. Furthermore we have the condition
\begin{align*}
    (\alpha_1+\alpha_2)k'=\alpha_2p, \quad \alpha_2pk'=\alpha_1m_1^2+\alpha_2(p^2+m_2^2)
\end{align*}
for the gradients. If $\alpha_1=-\alpha_2$ then $p=0$ and $m_1^2=m_2^2$ (a non-physical solution). Otherwise
\begin{align*}
    -m_1^2=k'^2=(1-\alpha_2)m_1^2+\alpha_2(p^2+m_2^2)
\end{align*}
and
\begin{align*}
    \alpha_2p^2=\frac{1}{2}(p^2-m_1^2+m_2^2)
\end{align*}
where we set $\alpha_1+\alpha_2=1$ without loss of generality. We obtain $-p^2=(m_1\pm m_2)^2$ and conclude that the physical Landau variety is
\begin{align*}
    L_\text{phys}(I)=\{(p,m_1,m_2)\in\mathbb{C}^6 \;|\; m_1=0\text{ or }m_2=0\text{ or }p^2=0\text{ or }-p^2=(m_1\pm m_2)^2\}.
\end{align*}
This is a classical result which can be found for example in \cite{s-matrix} (note that in the result there $-p^2$ is replaced by $p^2$ since we use a different convention for the Minkowski metric).

\subsubsection{Remark on the Ramification}
With the goal to eventually prove Cutkosky's Theorem in mind, we remark on the ramification of this integral around the points in the Landau variety. The integral (\ref{eq:onelooptwopoint}) is well-known and in kinematic renormalization at a renormalization point $p_0\in\mathbb{R}^4$ evaluates to (see for example \cite{kreimerintegral})
\begin{align}\label{eq:1loop2pointresult}
    I(p,m_1,m_2)=\frac{\sqrt{\lambda(-p^2,m_1^2,m_2^2)}}{-2p^2}\ln(\frac{m_1^2+m_2^2+p^2-\sqrt{\lambda(-p^2,m_1^2,m_2^2)}}{m_1^2+m_2^2+p^2+\sqrt{\lambda(-p^2,m_1^2,m_2^2)}})-\frac{m_1^2-m_2^2}{-2p^2}\ln(\frac{m_1^2}{m_2^2}) - (p\leftrightarrow p_0)
\end{align}
for $(m_1-m_2)<-p^2<(m_1+m_2)^2$ and $m_1^2,m_2^2>0$ where the branch cut of the logarithm $\ln(z)$ (not of the expression including the complicated argument of the logarithm in (\ref{eq:1loop2pointresult})) and the square root $\sqrt{z}$ are understood to lie on the positive real axis. Again we replaced $p^2$ in the expression from \cite{kreimerintegral} by $-p^2$ since we use a different convention for the Minkowski metric. Here $\lambda(a,b,c):=a^2+b^2+c^2-2(ab+bc+ca)$ is the Källén function. Note that the ramification of $I$ at the points $-p^2=(m_1\pm m_2)^2$ is somewhat subtle: The function $I$ does not have a singularity at these points (except when $p^2=m_1=m_2=0$) and thus the discontinuity does not directly arise from the essential singularity of the logarithm at 0. If we assume $m_1^2,m_2^2>0$, the Källén function factorizes as
\begin{align*}
    \lambda(-p^2,m_1^2,m_2^2)=(-p^2-(m_1+m_2)^2)(-p^2-(m_1-m_2)^2).
\end{align*}
In this form we immediately see that $-p^2=(m_1\pm m_2)^2$ are the zeros of $\lambda(-p^2,m_1^2,m_2^2)$. Computing the naive limit $p^2\to-(m_1\pm m_2)^2$ of $I(p,m1,m2)$ yields a finite result, the argument of the logarithm becomes 1. In fact the argument of the logarithm can only be 0 if one of the masses vanishes. The square-root is a multivalued function with two sheets and in particular $\sqrt{\lambda(-p^2,m_1^2,m_2^2)}$ changes the sign when continuing along a loop around a point with $-p^2=(m_1+m_2)^2$. So going around such a loop the first term in equation (\ref{eq:1loop2pointresult}) becomes
\begin{align*}
    -\frac{\sqrt{\lambda(p^2,m_1^2,m_2^2)}}{2p^2}\ln(\frac{m_1^2+m_2^2-p^2+\sqrt{\lambda(p^2,m_1^2,m_2^2)}}{m_1^2+m_2^2-p^2-\sqrt{\lambda(p^2,m_1^2,m_2^2)}}).
\end{align*}
To compare this to the expression we started with we use $\ln(\frac{1}{x})=-\ln(x)+2\pi i$ for $x\in\mathbb{R}^+$ and obtain
\begin{align*}
    \frac{\sqrt{\lambda(p^2,m_1^2,m_2^2)}}{2p^2}(\ln(\frac{m_1^2+m_2^2-p^2-\sqrt{\lambda(p^2,m_1^2,m_2^2)}}{m_1^2+m_2^2-p^2+\sqrt{\lambda(p^2,m_1^2,m_2^2)}})+2\pi i).
\end{align*}
So we cross the branch cut of the logarithm by crossing the branch cut of the square root.

\section{Conclusion and Outlook}
In this paper we showed how to apply the isotopy techniques developed in \cite{app-iso} to the case of quadratic integrals and in particular Feynman integrals in momentum space. From these techniques we obtain a criterion to test the analyticity of such an integral at a given point by means of a set of simple equations. Read in the finite chart, these are the classical Landau equations and at infinity we obtain a new set of equations for singularities of the second type.\\
It remains to show however that Feynman integrals are analytic at physical points on the principal branch outside of $L_\text{phys}(I_G)$. This can not be achieved by the isotopy techniques discussed in this work alone and requires new ideas. The main task is to show that every non-analytic point of a given graph $G$ arises as a solution $(\alpha,k,p,m)\in c(I_G)$ to the Landau equations with $\alpha\in\mathbb{R}_{\geq0}^{|E(G)|}-\{0\}$. An immediate consequence of this statement would be the absence of second type singularities on the principal branch.\\
In our discussion the $Q$-regularization method introduced in Subsection \ref{sec:regularization} proves to be very useful. It remains to hope that this technique can be applied to other problems relating to momentum space Feynman integrals, specifically Cutkosky's Theorem. Furthermore it might be worth studying the geometry of the intersection of the $Q$-regularized zero loci in its own right. A detailed discussion of the ramification of quadratic integrals around non-analytic points as well as dispersion relations is postponed for future research.

\section*{Acknowledgement}
I would like to express my gratitude towards Marko Berghoff for all the discussions and personal support, particularly during the lock-down due to the ongoing Corona crisis. Furthermore I thank Dirk Kreimer for his guidance and the opportunity to study the mysterious world of quantum field theory as well as Prof. John Collins for sharing his valuable insights on the topic with me. A special thanks goes to Olaf Müller for reassuring me and Raphael Kogler who had to endure my endless babbling and of course for proof-reading this (or rather a very old version of this) work.

\bibliography{ref}
\bibliographystyle{alpha} 

\end{document}